\documentclass[dvipsnames]{article}

\usepackage[preprint]{neurips_2023}

\usepackage[utf8]{inputenc} 
\usepackage[T1]{fontenc}    
\usepackage{hyperref}       
 \usepackage{url}            
\usepackage{booktabs}       
\usepackage{amsfonts}       
\usepackage{nicefrac}       
\usepackage{microtype} 
\usepackage{mathtools}
\usepackage{tikz}
\usetikzlibrary{arrows.meta}
\usetikzlibrary{decorations.pathreplacing}

\usepackage{xcolor}

\usepackage{subfiles}
\usepackage{times}
\usepackage{helvet}
\usepackage{courier}
\usepackage{graphicx}
\usepackage{natbib}
\usepackage{enumerate}
\usepackage{caption}
\usepackage{soul}
\usepackage[utf8]{inputenc}
\usepackage{amsmath}
\usepackage{amssymb}
\usepackage{amsthm}
\newtheorem{theorem}{Theorem}
\newtheorem{lemma}{Lemma}
\newtheorem{claim}{Claim}
\newtheorem{definition}{Definition}

\newtheorem{problem}{Problem}
\usepackage{algorithm}
\usepackage{float}
\usepackage[noend]{algpseudocode}
\usepackage{enumitem}
  \usepackage{subfigure}
\usepackage{nicefrac}
\usepackage{xspace}

\newcommand{\scpl}{Submodular Cover\xspace}
\newcommand{\scp}{SCP\xspace}
\newcommand{\mscp}{MSCP\xspace}
\newcommand{\msmp}{MSMP\xspace}
\newcommand{\rscp}{RSCP\xspace}
\newcommand{\smp}{SMP\xspace}

\newcommand{\argmin}{\text{argmin}}
\newcommand{\argmax}{\text{argmax}}
\newcommand{\distgreed}{\texttt{distorted-bi}\xspace}

\newcommand{\sg}{\texttt{stoch-bi}\xspace}
\newcommand{\sgc}{\texttt{stoch-greedy-c}\xspace}
\newcommand{\conv}{\texttt{convert}\xspace}
\newcommand{\convr}{\texttt{convert-rand}\xspace}
\newcommand{\tgc}{\texttt{thresh-greedy-c}\xspace}
\newcommand{\gc}{\texttt{greedy-c}\xspace}
\newcommand{\filt}{\texttt{stream-c}\xspace}
\newcommand{\reg}{\texttt{reg}\xspace}
\newcommand{\regconv}{\texttt{convert-reg}\xspace}
\newcommand{\mE}{\mathbb{E}}

\newcommand\explaineq[2]{\stackrel{\mathclap{\normalfont\mbox{#1}}}{#2}}

\author{%
  Wenjing Chen\\
  Department of Computer Science \& Engineering\\
  Texas A$\&$M University\\
  \And
  Victoria G. Crawford \\
  Department of Computer Science \& Engineering\\
  Texas A\&M University\\
}
\title{Bicriteria Approximation Algorithms for the Submodular Cover Problem}

\begin{document}


\maketitle

\begin{abstract}
    In this paper, we consider the optimization problem Submodular Cover (SCP), which is to find a minimum cardinality subset of a finite universe $U$ such that the value of a submodular function $f$ is above an input threshold $\tau$. In particular, we consider several variants of SCP including the general case, the case where $f$ is additionally assumed to be monotone, and finally the case where $f$ is a regularized monotone submodular function. Our most significant contributions are that: (i) We propose a scalable algorithm for monotone SCP that achieves nearly the same approximation guarantees as the standard greedy algorithm in significantly faster time; (ii) We are the first to develop an algorithm for general SCP that achieves a solution arbitrarily close to being feasible; and finally (iii) we are the first to develop algorithms for regularized SCP. Our algorithms are then demonstrated to be effective in an extensive experimental section on data summarization and graph cut, two applications of SCP.
\end{abstract}

\section{Introduction}
Submodularity captures a diminishing returns property of set functions: Let $f:2^U\to\mathbb{R}$ be defined over subsets of a universe $U$ of size $n$. Then $f$ is \textit{submodular} if for all $A\subseteq B\subseteq U$ and $x\notin B$, $f(A\cup\{x\})-f(A) \geq f(B\cup\{x\})-f(B)$. Examples of submodular set functions include cut functions in graphs \citep{balkanski2018non}, information-theoretic quantities like entropy and mutual information \citep{iyer2021generalized}, determinantal point processes \citep{gillenwater2012near}, and coverage functions \citep{bateni2017almost}. Submodular set functions arise in many important real-world applications including active learning \citep{kothawade2021similar,kothawade2022talisman}, partial label learning \citep{bao2022submodular}, structured pruning of neural networks \citep{el2022data}, data summarization \citep{tschiatschek2014learning}, and client selection in federated learning \citep{balakrishnan2022diverse}.

While the majority of existing work has focused on developing approximation algorithms to maximize a submodular function subject to some constraint \citep{nemhauser1978analysis,mirzasoleiman2015lazier,harshaw2019submodular,buchbinder2014submodular}, in this paper we focus on developing algorithms for the related optimization problem of \scpl (\scp), defined as follows.
\begin{problem}[\scpl (\scp)]
    Let $f:2^U\to\mathbb{R}_{\geq 0}$ be a nonnegative submodular set function defined over subsets of the ground set $U$ of size $n$. Given threshold $\tau$, \scp is to find $\argmin\{|X|:f(X)\geq\tau\}$ if such a set exists.
\end{problem}
\scp captures applications where we seek to achieve a certain value of $f$ in as few elements as possible. For example, consider data summarization, where a submodular function $f$ is formulated to measure how effectively a subset $X$ summarizes the entire dataset $U$ \citep{tschiatschek2014learning}. Then if we set $\tau=\max\{f(X):X\subseteq U\}$, \scp asks to find the set of minimum size in $U$ that achieves the maximum effectiveness as a summary. Another example is when expected advertising revenue functions are formulated over subsets of a social network \citep{hartline2008optimal}, then \scp asks how we can reach a certain amount of revenue while picking as small a subset of users as possible. 

In this paper, we propose and analyze algorithms for several variants of \scp including the general case, the case where $f$ is assumed to be monotone\footnote{A set function $f$ is monotone if for all $A\subseteq B\subseteq U$, $f(A)\leq f(B)$.}, and finally when $f$ is a regularized monotone submodular function (and potentially takes on negative values). In particular, the contributions of this paper are:
\begin{itemize}[noitemsep]
    \item[(i)] We first address the need for scalable algorithms for \scp where $f$ is assumed to be monotone. While the greedy algorithm finds the best possible approximation guarantee for monotone \scp (\mscp) \citep{feige1998}, it makes $O(n^2)$ queries of $f$ which may be impractical in many applications. We propose and introduce two algorithms for \mscp which achieve nearly the same theoretical guarantee as the greedy algorithm but only make $O(n\ln(n))$ queries of $f$. In addition, we extend the work of \cite{iyer2013submodular} to a method of converting fast randomized approximation algorithms for the dual cardinality constrained monotone submodular maximization problem (\msmp) into approximation algorithms for \mscp.
    \item[(ii)] Next, we address the need for algorithms that can produce nearly feasible solutions to the general \scp problem. In particular, we provide the first algorithm for \scp that, with input $\epsilon>0$, returns a solution $S$ that is guaranteed to satisfy: (i) $f(S)\geq(1-\epsilon)\tau$; and (ii) $|S|=O(|OPT|/\epsilon)$ where $OPT$ is an optimal solution to the instance. A caveat for our algorithm is that it is not necessarily polynomial time and requires an exact solution to \smp on an instance of size $O(|OPT|/\epsilon^2)$.
    \item[(iii)] Third, we are the first to consider \textit{regularized \scp} (\rscp). \rscp is where the objective $f=g-c$ where $g$ is a nonnegative, monotone, and submodular function and $c$ is a modular cost penalty function. $f$ is not necessarily monotone but potentially takes on negative values, and therefore this new problem doesn't fall under the general \scp problem. We develop a method of converting algorithms for the dual regularized submodular maximization problem \citep{harshaw2019submodular} into ones for \rscp. We then propose the first algorithm for \rscp, which is a greedy algorithm using queries to a distorted version of $f=g-c$. 
    \item[(iv)] Finally, we conduct an experimental analysis for our algorithms for \mscp and general \scp on instances of data summarization and graph cut. We find that our algorithms for \mscp makes a large speedup compared to the standard greedy approach, and we explore the pros and cons of each relative to the other. We also find that our algorithm for general \scp is practical for our applications despite not being guaranteed to run in polynomially many queries of $f$.
\end{itemize}
A table summarizing all of our algorithmic contributions can be found in the appendix. We now provide a number of preliminary definitions and notations that will be used throughout the paper.

\begin{subsection}{Preliminary definitions}

We first provide a number of prelimary definitions that will be used throughout the paper:
(i) The Submodular Maximization Problem (\smp) is the dual optimization problem to \scp defined by, given budget $\kappa$, find $\argmax\{f(X): X\subseteq U, |X|\leq\kappa\}$;
(ii) Monotone \scp (\mscp) is the version of \scp where $f$ is additionally assumed to be monotone;
(iii) Regularized \scp (\rscp) is a related problem to \scp where $f=g-c$ and $g$ is monotone, submodular, and nonnegative, while $c$ is a modular\footnote{Every $x\in U$ is assigned a cost $c_x$ such that $c(X)=\sum_{x\in X} c_x$.} nonnegative cost function;
(iv) $OPT$ is used to refer to the optimal solution to the instance of \scp that should be clear from the context;
(v) $OPT_{SM}$ is used to refer to the optimal solution to the instance of \smp that should be clear from the context;
(vi) An $(\alpha,\beta)$-bicriteria approximation algorithm for \scp returns a solution $X$ such that $|X|\leq\alpha |OPT|$ and $f(X)\geq\beta\tau$. An $(\alpha,\beta)$-bicriteria approximation algorithm for \smp returns a solution $X$ such that $f(X)\geq \alpha f(OPT)$ and $|X|\leq\beta\kappa$. Notice that the approximation on the objective is first, and the approximation on the constraint is second;
(vii) The marginal gain of adding an element $u\in U$ to  a set $S\subseteq U$ is denoted as $\Delta f(S,u)=f(S\cup u)-f(S)$;
(viii) The function $f_{\tau}=\min\{f,\tau\}$.
\end{subsection}

\subsection{Related Work}
\label{section:relatedwork}
\mscp is the most studied variant of \scp \citep{wolsey1982analysis,wan2010,mirzasoleiman2015,mirzasoleiman2016,crawford2019submodular}. The standard greedy algorithm produces a logarithmic approximation guarantee for \mscp in $O(n^2)$ queries of $f$ \citep{wolsey1982analysis}, and this is the best approximation guarantee that we can expect unless NP has $n^{\mathcal{O}(\log(\log(n)))}$-time deterministic algorithms \citep{feige1998}. One version of the greedy algorithm for \mscp works as follows: A set $S$ is initialized to be $\emptyset$. Iteratively, the element $\argmax\{\Delta f(S,x): x\in U\}$ is added to $S$ until $f(S)$ reaches $(1-\epsilon)\tau$. It has previously been shown that this is a $(\ln(1/\epsilon), 1-\epsilon)$-bicriteria approximation algorithm \citep{krause2008robust}. Beyond greedy algorithms, algorithms for the distributed setting
\citep{mirzasoleiman2015distributed,mirzasoleiman2016} as well as the streaming setting \citep{norouzi2016} for \mscp have been proposed.

On the other hand, developing algorithms for \scp in full generality is more difficult since monotonicity of $f$ is not assumed. The standard greedy algorithm does not have any non-trivial approximation guarantee for \scp. In fact, to the best of our knowledge, no greedy-like algorithms have been found to be very useful for \scp. Recently, \cite{crawford2023scalable} considered \scp and proved that it is not possible to develop an algorithm that guarantees $f(X)\geq\tau/2$ for \scp in polynomially many queries of $f$ assuming the value oracle model. On the other hand, algorithmic techniques that are used for \smp in the streaming setting \citep{alaluf2022optimal} proved to be useful for \scp. In particular, \cite{crawford2023scalable} proposed an algorithm using related techniques to that of \citeauthor{alaluf2022optimal} that achieves a $(O(1/\epsilon^2),1/2-\epsilon)$-bicriteria approximation guarantee for \scp in polynomially many queries of $f$. We also take an approach inspired by the streaming algorithm of \citeauthor{alaluf2022optimal}, but sacrifice efficiency in order to find a solution for \scp that is arbitrarily close to being feasible.

\smp is the dual optimization problem to \scp, and has received relatively more attention than \scp \citep{nemhauser1978,badanidiyuru2014fast,mirzasoleiman2015lazier, feige2011maximizing, buchbinder2014submodular,alaluf2022optimal}. \cite{iyer2013} proposed a method of converting algorithms for \smp to ones for \scp. In particular, given a deterministic $(\gamma,\beta)$-bicriteria approximation algorithm for \smp, the algorithm \conv (see pseudocode in the appendix) proposed by \citeauthor{iyer2013} produces a deterministic $((1+\alpha)\beta,\gamma)$-bicriteria approximation algorithm for \scp. The algorithm works by making $\log_{1+\alpha}(n)$ guesses for $|OPT|$ (which is unknown in \scp), running the \smp algorithm with the budget set to each guess, and returning the smallest solution with $f$ value above $\gamma\tau$. However, this approach is limited by the approximation guarantees of existing algorithms for \smp. The best $\gamma$ for monotone \smp is $1-1/e$, and the best for general \smp where $f$ is not assumed to be monotone is significantly lower \citep{gharan2011submodular}. Several of the algorithms that we propose in this paper do generally follow the model of \conv in that they rely on guesses of $|OPT|$, but are different because they: (i) Implicitly use bicriteria approximation algorithms for \smp which have better guarantees on the objective ($\gamma$) because they do not necessarily return a feasible solution; (ii) Are more efficient with respect to the number of queries of $f$, since \conv potentially wastes many queries of $f$ by doing essentially the same behavior for different guesses of $|OPT|$.


\begin{section}{Algorithms and Theoretical Guarantees}
\label{section:theoretical}
    In this section, we present and theoretically analyze our algorithms for several variants of \scp. 
    In particular, in Section \ref{section:monotone_objectives} we first consider \mscp. We present a method of converting randomized algorithms for \smp to algorithms for \scp, and then we present the algorithms \tgc and \sgc for \mscp, which both have lower query complexity compared to the standard greedy algorithm. Next, we consider the general problem of \scp in Section \ref{section:nonmono}. We present the algorithm \filt for \scp, which produces a solution with $f$ value arbitrarily close to $\tau$, but does not necessarily make polynomially many queries to $f$. Finally, in Section \ref{section:reg}, we consider \rscp. We present a method of converting algorithms for regularized \smp to ones for \rscp, and then introduce the algorithm \distgreed for \rscp.
 
\begin{subsection}{Monotone submodular cover}
\label{section:monotone_objectives}
In this section, we develop and analyze approximation algorithms for \mscp. The greedy algorithm is a tight $(\ln(1/\epsilon),1-\epsilon)$-bicriteria approximation algorithm for \mscp \citep{krause2008robust}. However, the greedy algorithm makes $O(n^2)$ queries of $f$, which is impractical in many application settings with large $U$ and/or when queries of $f$ are costly \citep{mirzasoleiman2015lazier}. Motivated by this, we propose and analyze the algorithms \tgc and \sgc for \mscp which give about the same bicriteria approximation guarantees but in many fewer queries of $f$.

We first describe \tgc. \tgc is closely related to the existing threshold greedy algorithm for monotone \smp \citep{badanidiyuru2014fast}, and therefore we relegate the pseudocode of \tgc to the appendix and only include a brief discussion here. At each iteration of \tgc, instead of picking the element with highest marginal gain into $S$, it adds all elements in $U$ with marginal gain above a threshold, $w$. $w$ is initialized to $\max_{u\in U}f(\{u\})$, and is decreased by a factor of $(1-\epsilon/2)$ when the algorithm proceeds to the next iteration. \tgc adds elements to a solution $S$ until $f(S)$ reaches $(1-\epsilon)\tau$, which is shown to happen in at most 
$\ln(2/\epsilon)|OPT|+1$ elements in the proof of Theorem \ref{thm:threshold}. We now state the theoretical guarantees of \tgc in Theorem \ref{thm:threshold}.
\begin{theorem}
\label{thm:threshold}
 \tgc produces a solution with $(\ln(2/\epsilon)+1,1-\epsilon)$-bicriteria approximation guarantee to MSCP, in $O(\frac{n}{\epsilon}\log(\frac{n}{\epsilon}))$ number of queries of $f$.
\end{theorem}

Another method of speeding up the standard greedy algorithm is by introducing randomization, as has been done for monotone \smp  in the stochastic greedy algorithm \citep{mirzasoleiman2015lazier}. A natural question is whether a randomized algorithm for monotone \smp can be converted into an algorithm for \mscp using the algorithm \conv of \citeauthor{iyer2013}. However, \conv relies on a deterministic approximation guarantee. We now introduce a new algorithm called \convr that is analogous to \conv but runs the \smp algorithm $O(\ln(n)\ln(1/\delta))$ times in order to have the approximation guarantee hold with high probability. Pseudocode for \convr, as well as a proof of Theorem \ref{thm:convert} can be found in the appendix.
\begin{theorem}
    \label{thm:convert}
    Any randomized $(1-\epsilon/2,\beta)$-bicriteria approximation algorithm for monotone \smp that runs in time $\mathcal{T}(n)$ where $\gamma$ holds only in expectation can be converted into a $((1+\alpha)\beta,1-\epsilon)$-bicriteria approximation algorithm for \mscp that runs in time $O(\log_{1+\alpha}(|OPT|)\ln(1/\delta)\mathcal{T}(n))$ where $\gamma$ holds with probability at least $1-\delta$.
\end{theorem}
Therefore by applying Theorem \ref{thm:convert} to the stochastic greedy algorithm of \citeauthor{mirzasoleiman2015lazier}, we have a $(1+\alpha,1-1/e-\epsilon)$-bicriteria approximation algorithm for \mscp with high probability in $O(n\log_{1+\alpha}(|OPT|)\ln(1/\delta)\ln(1/\epsilon))$ queries of $f$.
However, a factor of $1-1/e-\epsilon$ of $\tau$ is not very close to feasible, and further the \convr method wastes many queries of $f$ essentially doing the same computations for different guesses of $|OPT|$. Therefore we focus the rest of this section on developing an algorithm, \sgc, that uses the techniques of the stochastic greedy algorithm more directly for \mscp.

The idea behind the stochastic greedy algorithm for \smp is that instead of computing the marginal gains of all elements at each iteration, we take a uniformly random sampled subset from $U$ and pick the element with the highest marginal gain among the sampled subset. If the sampled subset is sufficiently large, in particular of size at least $(n/\kappa)\ln(1/\epsilon)$ where $\kappa$ is the budget for the instance of \smp, then with high probability a uniformly random element of $OPT_{SM}$ will appear in the sampled subset and the marginal gain of adding the element is nearly the same as the standard greedy algorithm in expectation. However, in \smp we know that $|OPT_{SM}|=\kappa$, but in \mscp $|OPT|$ is unknown. Therefore it is not obvious how to apply this technique in a more direct way than \convr.

We now introduce our algorithm \sgc for \mscp, pseudocode for which is provided in Algorithm \ref{alg:stoc_cover}. \sgc takes as input $\epsilon >0$, $\delta >0$, $\alpha > 0$, and an instance of \mscp. \sgc keeps track of $O(\ln(1/\delta))$ possibly overlapping solutions $S_1,S_2,...$ throughout a sequence of iterations. \sgc also keeps track of an estimate of $|OPT|$, $g$. During each iteration, for each solution $S_i$, \sgc uniformly randomly and independently samples a set $R$ of size $\min\{n,(n/g)\ln(3/\epsilon)\}$ and adds $u=\argmax\{\Delta f_{\tau}(S_i, x): x\in R\}$ to $S_i$. Every time $\frac{\alpha}{1+\alpha}\ln(3/\epsilon)g$ elements have been added to each $S_i$, $g$ is increased by a factor of $1+\alpha$. \sgc stops once there exists an $S_i$ such that $f(S_i)\geq (1-\epsilon)\tau$, and returns this solution.

\begin{algorithm}[t]
\caption{\sgc}\label{alg:stoc_cover}
\textbf{Input}: $\epsilon,\alpha,\tau$\\
\textbf{Output}: $S\subseteq U$
\begin{algorithmic}[1]
    \State $S_i\gets \emptyset$ $\forall i\in\left\{1,...,\ln(1/\delta)/\ln(2)\right\}$
  \State $r\gets 1, g\gets 1+\alpha$
  \While {$f(S_i) < (1-\epsilon)\tau$ $\forall i$}
    \For {$i\in\left\{1,...,\ln(1/\delta)/\ln(2)\right\}$} \label{sgc:iter_start}
        \State \label{line:sample_R} $R\gets$ sample $\min\{n, n\ln(3/\epsilon)/g\}$ elements from $U$
        \State \label{sgc:add_new}$u\gets\argmax_{x\in R}\Delta f_\tau(S_{j},x)$\label{line:stochadd}
        \State $S_j\gets S_j\cup\{u\}$
    \EndFor
    \State{$r\gets r+1$}
    \If{$r > \ln(3/\epsilon)g$} $g\gets (1+\alpha)g$\EndIf
  \EndWhile
  \State \textbf{return} $\argmin\{|S_i|: f(S_i)\geq(1-\epsilon)\tau\}$\label{line:complete}
\end{algorithmic}
\end{algorithm}

We now state the theoretical results for \sgc in Theorem \ref{theorem:stochastic}.
\begin{theorem}
    \label{theorem:stochastic}
    Suppose that \sgc is run for an instance of \mscp.
   Then with probability at least $1-\delta$, \sgc
   outputs a solution $S$ that satisfies a $((1+\alpha)\lceil\ln(3/\epsilon)\rceil,1-\epsilon)$-bicriteria approximation guarantee in at most $O\left(\frac{\alpha}{1+\alpha}n\ln(1/\delta)\ln^2(3/\epsilon)\log_{1+\alpha}(|OPT|)\right)$ queries of $f$.
\end{theorem}
Compared to \tgc, \sgc has a better dependence on $\epsilon$ in terms of the number of queries made to $f$. In addition, it is possible to extend the stochastic greedy algorithm of \citeauthor{mirzasoleiman2015lazier} to a $(1-\epsilon, \ln(1/\epsilon)$-bicriteria approximation algorithm for \smp and then use \conv (see the appendix). However, \sgc still would have strictly fewer queries of $f$ by a factor of $\frac{\alpha}{1+\alpha}$ compared to this approach because \conv does essentially the same computations for different guesses of $|OPT|$. In order to prove Theorem \ref{theorem:stochastic}, we first need Lemma \ref{lemma:marggain} below, which states that as long as $g\leq (1+\alpha)|OPT|$, the marginal gain of adding $u$ in Line \ref{line:stochadd} is about the same as the standard greedy algorithm in expectation.
\begin{lemma}
    \label{lemma:marggain}
    Consider any of the sets $S_i$ at the beginning of an iteration on Line \ref{sgc:iter_start} where $g\leq (1+\alpha)|OPT|$.
    Then if $u_i$ is the random element that will be added on Line \ref{sgc:add_new}, we have that  $\mathbb{E}[\Delta f_{\tau}(S_i, u_i)]\geq \frac{1-\epsilon/3}{(1+\alpha)|OPT|}(\tau - f_{\tau}(S_i))$.
\end{lemma}
Further, Lemma \ref{lem:sgc_expection} below uses Lemma \ref{lemma:marggain} to show that by the time $g$ reaches $(1+\alpha)|OPT|$, $\mathbb{E}[f_{\tau}(S_i)]\geq (1-\frac{\epsilon}{2})\tau$ for all $i$.
\begin{lemma}
\label{lem:sgc_expection}
    Once $r$ reaches $(1+\alpha)\lceil\ln(3/\epsilon)|OPT|\rceil$, we have that $\mathbb{E}[f_{\tau}(S_i)] \geq\left(1-\frac{\epsilon}{2}\right)\tau$ for all $i$.
\end{lemma}
Finally, because there are $O(\ln(1/\delta))$ solutions, by the time $g$ reaches $(1+\alpha)|OPT|$, there exists $i$ such that $f(S_i)\geq (1-\epsilon)\tau$ with probability at least $1-\delta$ by using concentration bounds, which is stated in Lemma \ref{lem:sgc_high_prob}.
\begin{lemma}
\label{lem:sgc_high_prob}
    With probability at least $1-\delta$, once $r$ reaches $(1+\alpha)\lceil\ln(3/\epsilon)|OPT|\rceil$, we have that $\max_{i}f(S_i)\geq(1-\epsilon)\tau$.
\end{lemma}
Lemma \ref{lem:sgc_high_prob} allows us to keep increasing $g$ by a factor of $(1+\alpha)$ periodically, because intuitively the longer we keep adding elements, the bigger we know that $|OPT|$ must be since the algorithm is still running and none of the solution sets has reached $(1-\epsilon)\tau$ yet. The proof of Lemmas \ref{lemma:marggain} and \ref{lem:sgc_expection}, and of Theorem \ref{theorem:stochastic} can be found in the appendix.

\end{subsection}
 \begin{subsection}{Non-monotone submodular cover}
\label{section:nonmono}
In this section, we introduce and theoretically analyze the algorithm \filt for \scp in the general setting, where $f$ is not assumed to be monotone. In the general setting, the standard greedy algorithm doesn't have non-trivial approximation guarantee for \scp. In addition, it has previously been shown that it is not possible for an algorithm to guarantee that $f(X)\geq\tau/2$ for \scp, where $X$ is its returned solution, in polynomially many queries of $f$ assuming the value oracle model \citep{crawford2023scalable}. Our algorithm \filt \textit{does} produce a solution $X$ that is guaranteed to satisfy $f(X)\geq (1-\epsilon)\tau$, but relies on solving an instance of \smp exactly on a set of size $O(|OPT|/\epsilon^2)$. Despite not being polynomial time, \filt is practical in many instances of \scp because: (i) $|OPT|$ may be relatively small; and (ii) the instance of \smp may be relatively easy to solve, e.g. $f$ may be very close to monotone on the instance of \smp even if it was very non-monotone on the original instance of \scp. These aspects of \filt are further explored in Section \ref{section:exp}.

We now describe \filt, pseudocode for which can be found in Algorithm \ref{alg:filter}. \filt takes as input $\epsilon >0$, $\alpha > 0$, and an instance of \scp. \filt takes sequential passes through the universe $U$ (Line \ref{line:loop}) with each pass corresponds to a new guess of $|OPT|$, $g$. $g$ is initialized as $1+\alpha$, and at the end of each pass is increased by a factor of $1+\alpha$. Throughout \filt, a subset of elements of $U$ are stored into $2/\epsilon$ disjoint sets, $S_1,...,S_{2/\epsilon}$. An element $u$ is stored in at most one set $S_j$ if both of the following are true: (i) $|S_j|<2g/\epsilon$; (ii) adding $u$ is sufficiently beneficial to increasing the $f$ value of $S_j$ i.e. $\Delta f(S_j,u)\geq\epsilon\tau/(2g)$. If no such $S_j$ exists, $u$ is discarded. At the end of each pass, \filt finds $S=\argmax\{f(X):X\subseteq\cup S_i, |X|\leq2g/\epsilon\}$ on Line \ref{line:max}. If $f(S)\geq(1-\epsilon)\tau$, then $S$ is returned and \filt terminates. 

\begin{algorithm}[t]
\caption{\filt}\label{alg:filter}
\textbf{Input}: $\epsilon$, $\alpha$\\
\textbf{Output}: $S\subseteq U$
\begin{algorithmic}[1]
  \State $S\gets\emptyset,S_1\gets\emptyset,...,S_{2/\epsilon}\gets\emptyset$
  \State $g\gets 1+\alpha$
  \While {$f(S)<(1-\epsilon)\tau$}\label{line:filter_stop}
  \For {$u\in U$}\label{line:loop}
    \If {$\exists j$ s.t.
         $\Delta f(S_j,u)\geq \epsilon \tau/(2 g)$ and $|S_j| < 2 g/\epsilon$}\label{line:add}
         \State $S_j\gets S_j\cup\{u\}$
    \EndIf
  \EndFor
  \State \label{alg:exact_step} $S\gets\argmax\{f(X): X\subseteq \cup_{i=1}^{2/\epsilon}S_i, |X|\leq 2g/\epsilon\}$\label{line:max}
  \State $g=(1+\alpha)g$
  \EndWhile
  \State \textbf{return} $S$
\end{algorithmic}
\end{algorithm}

We now present the theoretical guarantees of \filt in Theorem \ref{theorem:filt}.
\begin{theorem}
  \label{theorem:filt}
  Suppose that \filt is run for an instance of \scp. Then \filt returns $S$ such that $f(S)\geq (1-\epsilon)\tau$ and $|S|\leq (1+\alpha)(2/\epsilon)|OPT|$ in at most
  $$\log_{1+\alpha}(|OPT|)
    \left(\frac{2n}{\epsilon}+
    \mathcal{T}\left((1+\alpha)\left(\frac{4}{\epsilon^2}|OPT|\right)\right)\right)$$
    queries of $f$, where $\mathcal{T}(m)$ is the number of queries to $f$ of the algorithm for \smp used on Line \ref{line:max} of Algorithm \ref{alg:filter} on an input set of size $m$.
\end{theorem}

The key idea for proving Theorem \ref{theorem:filt} is that by the time $g$ is in the region $[|OPT|, (1+\alpha)|OPT|]$, there exists a subset $X\subseteq\cup S_i$ such that $|X|\leq 2g/\epsilon$ and $f(X)\geq (1-\epsilon)\tau$. In fact, it is shown in the proof of Lemma \ref{lemma:filt} in the appendix that the set $X$ is $S_t\cup (\cup_{i} S_i\cap OPT)$ for a certain one of the sets $S_t$. Then when we solve the instance of \smp on Line \ref{line:max}, we find a set that has these same properties as $X$, and \filt returns this set and terminates. Because $g\leq (1+\alpha)|OPT|$, the properties described in Theorem \ref{theorem:filt} hold. Further notice that $|\cup_{i}S_i|\leq 2(1+\alpha)|OPT|/\epsilon^2$ at all times before \filt exits, which implies the bounded query complexity in Theorem \ref{theorem:filt}. The key idea for proving Theorem \ref{theorem:filt} is stated below in Lemma \ref{lemma:filt} and proven in the appendix.
\begin{lemma}
    \label{lemma:filt}
    By the time that $g$ reaches the region $[|OPT|, (1+\alpha)|OPT|]$ and the loop on Line \ref{line:loop} of \filt has completed, there exists a set $X\subseteq \cup S_i$ of size at most $2(1+\alpha)|OPT|/\epsilon$ such that
    $f(X)\geq (1-\epsilon)\tau$.
\end{lemma}


\end{subsection}

 \begin{subsection}{Regularized monotone submodular cover}
\label{section:reg}
The final class of submodular functions we consider take the form $f=g-c$ where $g$ is monotone, submodular, and nonnegative, while $c$ is a modular, nonnegative penalty cost function, called \rscp. In this case, $f$ may take on negative values and therefore this class of submodular functions does not fit into general \scp. $f$ may also be nonmonotone. Existing theoretical guarantees for the dual problem of regularized \smp are in a different form than typical approximation algorithms \citep{harshaw2019submodular,kazemi2021regularized}, which we will describe in more detail below, and as a result \conv cannot be used. Motivated by this, we first develop an algorithm, \regconv, that takes algorithms for regularized \smp and converts them into an algorithm for \rscp. Next, we propose a generalization of the distorted greedy algorithm of \citeauthor{harshaw2019submodular} for regularized \smp, called \distgreed, that can be used along with \regconv to produce an algorithm for \rscp.

Existing proposed algorithms for regularized \smp have guarantees of the following form: Given budget $\kappa$, the regularized \smp algorithm is guaranteed to return a set $S$ such that $|S|\leq\kappa$ and $g(S)-c(S)\geq \gamma g(OPT_{SM})-c(OPT_{SM})$ where $\gamma$ is some value less than 1, e.g. $1-1/e$ for the distorted greedy algorithm of \citeauthor{harshaw2019submodular}. A guarantee of this form means \conv cannot be used (the check on Line \ref{line:convtest} of the pseudocode for \conv in the appendix is the problem). Motivated by this, we provide \regconv for these different types of theoretical guarantees.


We now describe \regconv, pseudocode for which can be found in Algorithm \ref{alg:regconv}. \regconv takes as input an algorithm \reg for regularized \smp, and $\alpha > 0$. \regconv repeatedly makes guesses for $|OPT|$, $\kappa$. For each guess $\kappa$, the algorithm \reg is run on an instance of \smp with objective $g-(\gamma/\beta)c$ and budget $\kappa$. Once $g-(\gamma/\beta)c$ reaches $\gamma\tau$, \regconv exits.

\begin{algorithm}[t]
\caption{\regconv}
\label{alg:regconv}
\textbf{Input}: $\alpha > 0$\\
\textbf{Output}: $S\subseteq U$
\begin{algorithmic}[1]
    \State $\kappa\gets 1+\alpha$, $S\gets\emptyset$
    \While{$g(S)-\frac{\gamma}{\beta}c(S)<\gamma\tau$}
        \State $S\gets$\reg run with objective $g-\frac{\gamma}{\beta}c$ and budget $\kappa$
        \State $\kappa\gets (1+\alpha)\kappa$
    \EndWhile
    \State \textbf{return} $S$
\end{algorithmic}
\end{algorithm}

The theoretical guarantees of \regconv are stated below in Theorem \ref{thm:rconv} and proven in the appendix. Theorem \ref{thm:rconv} makes a slightly stronger assumption on \reg than its approximation guarantees relative to $OPT_{SM}$. In particular, it is assumed that it returns a solution satisfying $|S|\leq\rho\kappa$ and $g(S)-c(S)\geq \gamma g(X)-c(X)$ \textit{for all $X\subseteq U$ such that $|X|\leq\kappa$}, not just for $OPT_{SM}$. However, this is true of many algorithms for regularized \smp including the distorted greedy algorithm of \citeauthor{harshaw2019submodular}. The idea behind running \reg with the objective $g-\frac{\gamma}{\beta}c$ instead of the actual objective $g-c$, is that
by the time $\kappa$ is a good guess of $|OPT|$, it is shown in the proof of Theorem \ref{thm:rconv} that with this different objective $g(S)-\frac{\gamma}{\beta}c(S)\geq\gamma\tau$.

\begin{theorem}
    \label{thm:rconv}
    Suppose that we have an algorithm \reg for regularized \smp, and given budget $\kappa$ \reg is guaranteed to return a set $S$ of cardinality at most $\rho\kappa$ such that
    $g(S)-c(S)\geq \gamma g(X)-\beta c(X)$ for all $X$ such that $|X|\leq\kappa$, in time $T(n)$.
    Then the algorithm \regconv using \reg as a subroutine returns a set $S$ in time $O(\log_{1+\alpha}(n)T(n))$ such that
    $|S|\leq (1+\alpha)\rho|OPT|$ and 
    $g(S)-\frac{\gamma}{\beta}c(S)\geq \gamma\tau.$
\end{theorem}
If we use \regconv on the distorted greedy algorithm of \citeauthor{harshaw2019submodular}, we end up with an algorithm for \rscp that is guaranteed to return a set $S$ such that $|S|\leq (1+\alpha)|OPT|$ and $g(S)-(1-1/e)c(S)\geq (1-1/e)\tau$. If we set $c=0$, then the problem setting reduces to \mscp and the distorted greedy algorithm of \cite{harshaw2019submodular} is equivalent to the standard greedy algorithm. However, our approximation guarantee does not reduce to the $(\ln(1/\epsilon),1-\epsilon)$-bicriteria approximation guarantee that would be preferable. A more intuitive result would be one that converges to that of the standard greedy algorithm as $c$ goes to 0. Motivated by this, we now propose an extension of the distorted greedy algorithm of \cite{harshaw2019submodular} for regularized \smp, \distgreed, that accomplishes this.

We now describe \distgreed, pseudocode for which can be found in the appendix. \distgreed takes as input an instance of regularized \smp and $\epsilon>0$. \distgreed is related to the standard greedy algorithm, but instead of making queries to $g-c$, \distgreed queries a distorted version of $g-c$ that de-emphasizes $g$ compared to $c$, and evolves over time. In particular, when element $i$ is being added to the solution set, we choose the element of maximum marginal gain, provided it is positive, to the objective $$\Phi_i(X) = \left(1-\frac{1}{\kappa}\right)^{\ln(1/\epsilon)\kappa-i}g(X) - c(X).$$

The theoretical guarantees of \distgreed are now presented in Theorem \ref{distthm}, and the proof of Theorem \ref{distthm} can be found in the appendix.
\begin{theorem}
  \label{distthm}
  Suppose that \distgreed is run for an instance of regularized \smp. Then
  \distgreed produces a solution $S$ in $O(n\kappa\ln(1/\epsilon))$ queries of $f$ such that $|S|\leq\ln(1/\epsilon)\kappa$ and for all $X\subseteq U$ such that $|X|\leq\kappa$,
  $g(S)-c(S)\geq (1-\epsilon)g(X)-\ln(1/\epsilon)c(X).$
\end{theorem}

Therefore by running \regconv with \distgreed as a subroutine for regularized \smp, we end up with an algorithm for regularized \rscp that is guaranteed to return a set $S$ such that $|S|\leq (1+\alpha)\ln(1/\epsilon)|OPT|$ and $g(S)-(1-\epsilon)c(S)/\ln(1/\epsilon)\geq (1-\epsilon)\tau$ in $O((1+\alpha)n|OPT|\log_{1+\alpha}(n)\log(1/\epsilon))$ queries of $f$. 

\end{subsection}
\end{section}
\section{Experiments}
\label{section:exp}
In this section, we experimentally evaluate the algorithms proposed in Sections \ref{section:monotone_objectives} and \ref{section:nonmono} on real instances of \scp. In particular, the emphasis of Section \ref{sec:experiments_mono} is on evaluation of our algorithm \sgc on instances of data summarization, an application of \mscp. Next, we evaluate \filt on instances of graph cut, an application of \scp that is not monotone, in Section \ref{sec:experiments_nonmono}. Additional details about the applications, setup, and results can be found in the appendix.
\subsection{Monotone submodular objective}
\label{sec:experiments_mono}
We first compare the solutions returned by \sgc ("SG"), \gc ("G"), \tgc ("TG"), and \conv using the \textit{bicriteria} extension of the stochastic greedy algorithm of \citeauthor{mirzasoleiman2015lazier} (see appendix) ("SG2") on instances of data summarization. The data summarization instance featured here in the main paper is the delicious dataset of URLs tagged with topics, and $f$ takes a subset of URLs to the number of distinct topics represented by those URLs ($n=5000$ with $8356$ tags) \citep{soleimani2016semi}. Additional datasets are explored in the appendix. We run the algorithms with input $\epsilon$ in the range $(0,0.15)$ and threshold values between 0 and $f(U)$ ($f(U)$ is the total number of tags). When $\epsilon$ is varied, $\tau$ is fixed at $0.6f(U)$. When $\tau$ is varied, $\epsilon$ is fixed at $0.2$. The parameter $\alpha$ is set to be $0.1$ and the initial guess of $|OPT|$ for \sgc and \conv is set to be $\tau/\max_{s}f(s)$.

The results in terms of the $f$ values and size of the solutions are presented in Figure \ref{fig:cover5000_tau_f} and \ref{fig:cover5000_tau_c}. From the plots, one can see that the $f$ values and size of solutions returned by \sgc, \gc, \tgc are nearly the same, and are smaller than the ones returned by \conv. This is unsurprising, because the theoretical guarantees on $f$ and size are about the same for the different algorithms, but \conv tends to perform closer to its worst case guarantee on size.
The number of queries to $f$ for different $\epsilon$ and $\tau$ are depicted in Figures \ref{fig:cover5000_eps_q} and \ref{fig:cover5000_tau_q}. Recall that the theoretical worst case number of queries to $f$ for \sgc, \gc, \tgc and \conv are $O((\alpha/(1+\alpha))n\ln^2(1/\epsilon)\log_{1+\alpha}(|OPT|))$, $O(n\ln(1/\epsilon)|OPT|)$, $O(n\log(|OPT|/\epsilon)/\epsilon)$, and $O(n\ln^2(1/\epsilon)\log_{1+\alpha}(|OPT|)$ respectively. As expected based on these theoretical guarantees, \gc does the worst and increases rapidly as $\tau$ (and therefore $|OPT|$) increases. \tgc tends to do worse compared to \sgc and \conv as $\epsilon$ gets smaller. \sgc consistently performs the fastest out of all of the algorithms.

\begin{figure*}[t!]
    \centering
    \hspace{-0.5em}
     \subfigure[cover $f$]
{\label{fig:cover5000_tau_f}\includegraphics[width=0.24\textwidth]{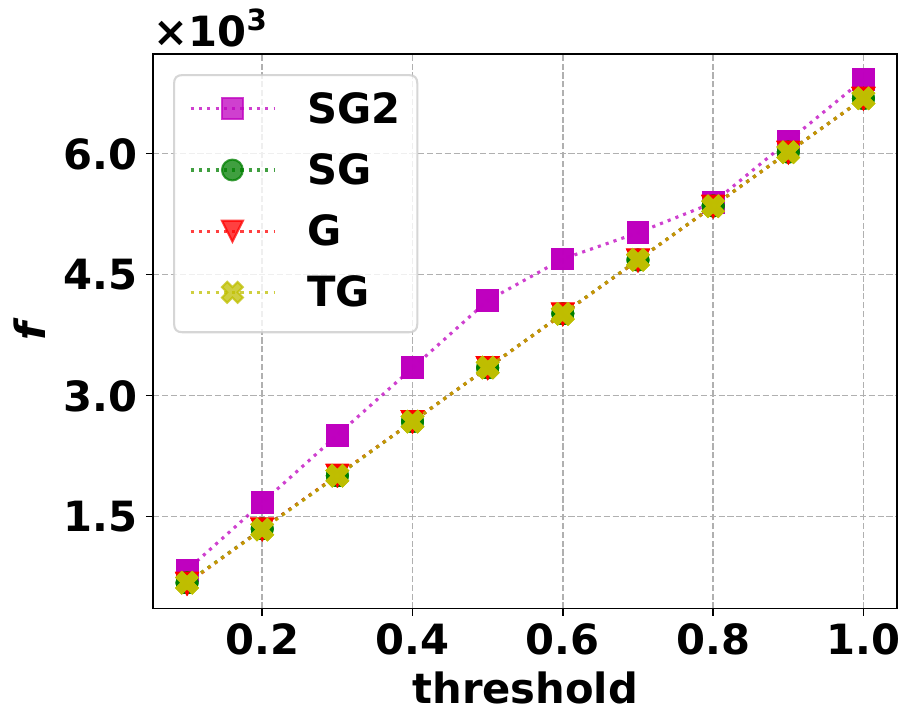}} 
\hspace{-0.5em}
     \subfigure[cover $c$]
{\label{fig:cover5000_tau_c}\includegraphics[width=0.24\textwidth]{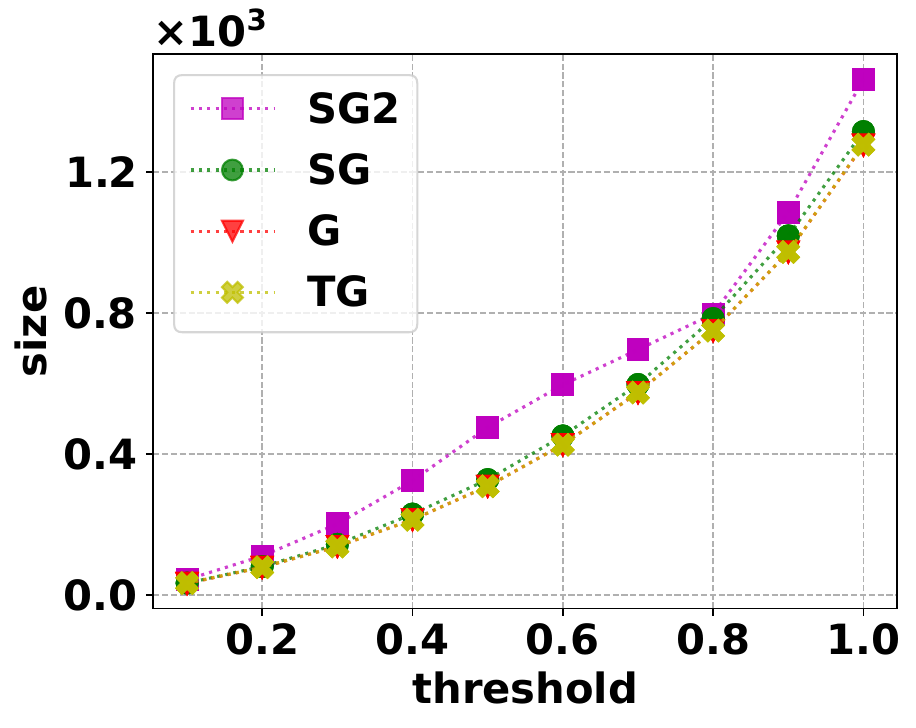}} 
    \hspace{-0.5em}
     \subfigure[cover queries]
{\label{fig:cover5000_tau_q}\includegraphics[width=0.24\textwidth]{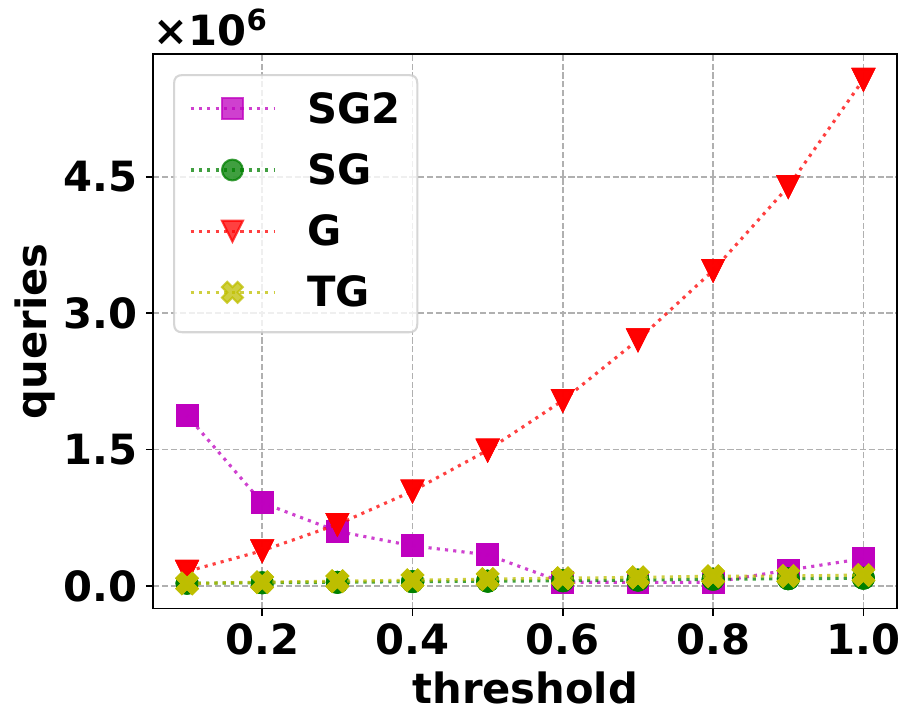}} 
\hspace{-0.5em}
\subfigure[cover queries]
{\label{fig:cover5000_eps_q}\includegraphics[width=0.24\textwidth]{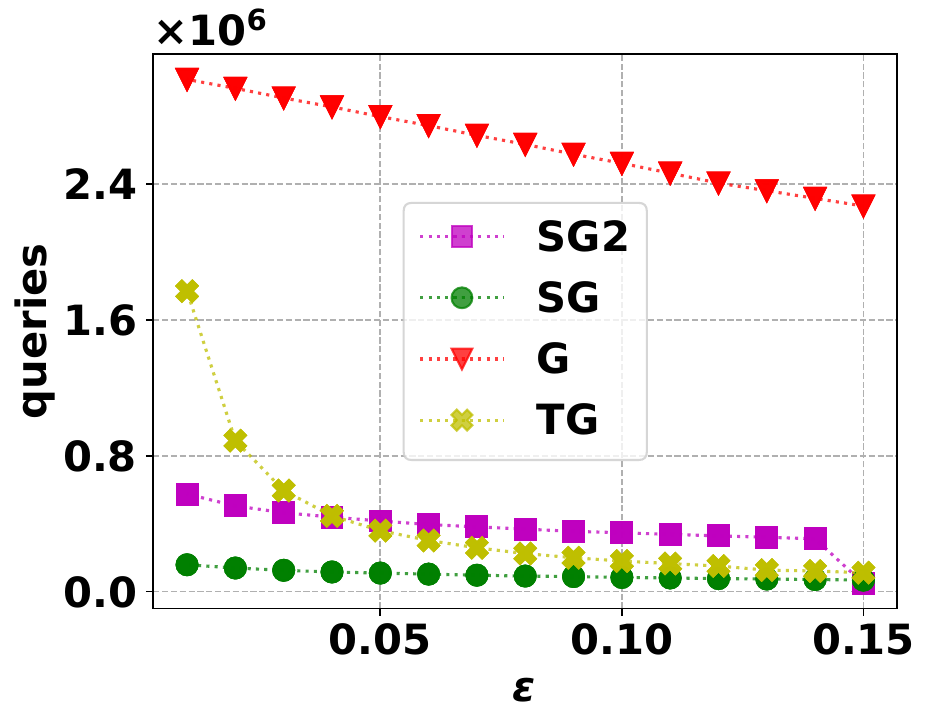}}
\hspace{-0.5em}
\subfigure[euall $f$]
{\label{fig:euall_tau_f}\includegraphics[width=0.24\textwidth]{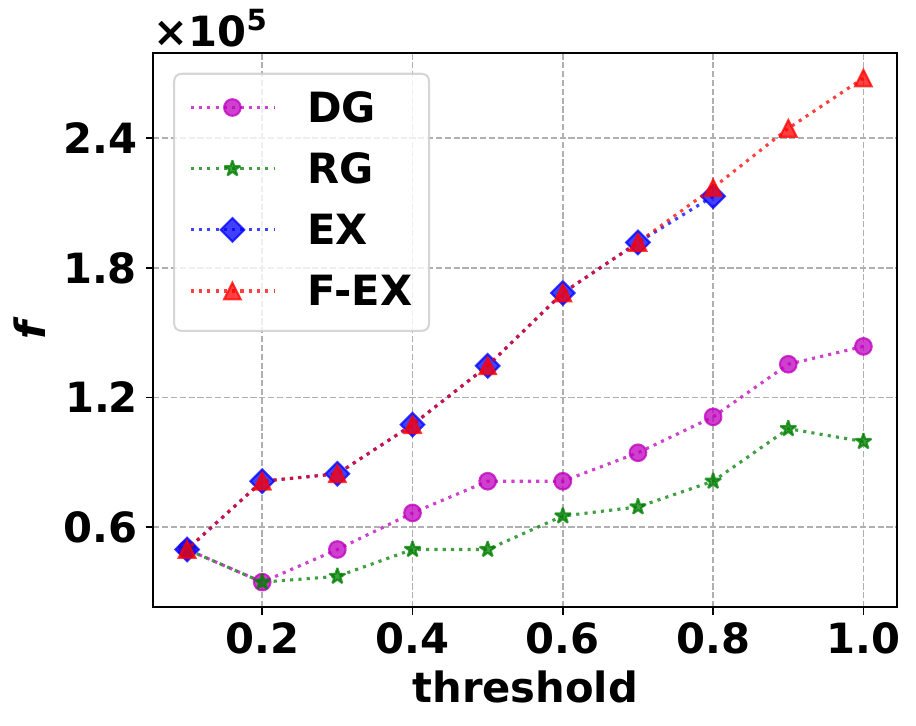}}
\hspace{-0.5em}
\subfigure[euall $c$]
{\label{fig:euall_tau_c}\includegraphics[width=0.24\textwidth]{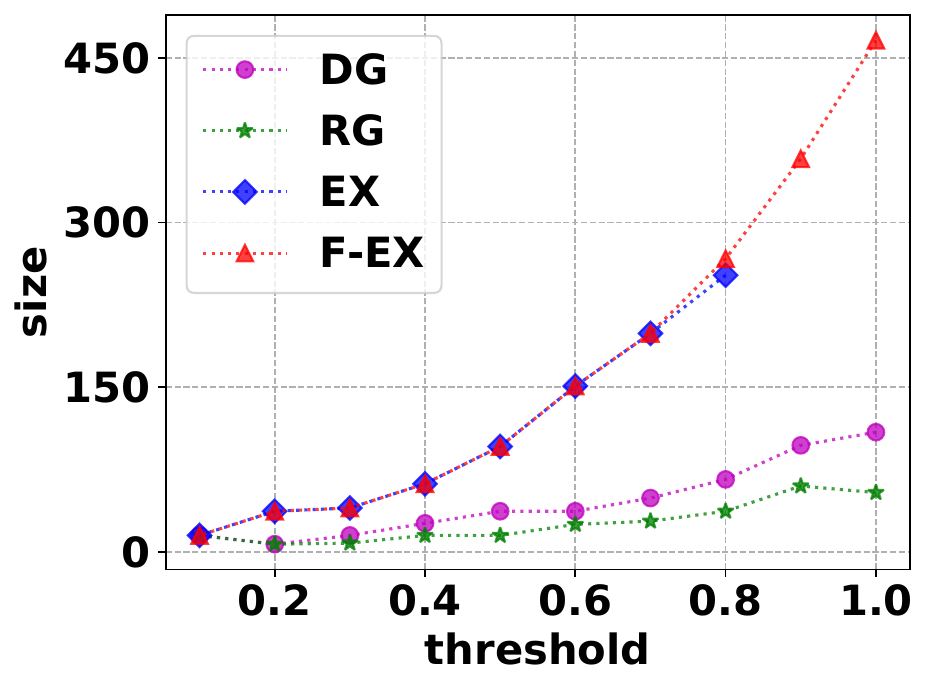}}
\hspace{-0.5em}
\subfigure[euall queries]
{\label{fig:euall_tau_q}\includegraphics[width=0.24\textwidth]{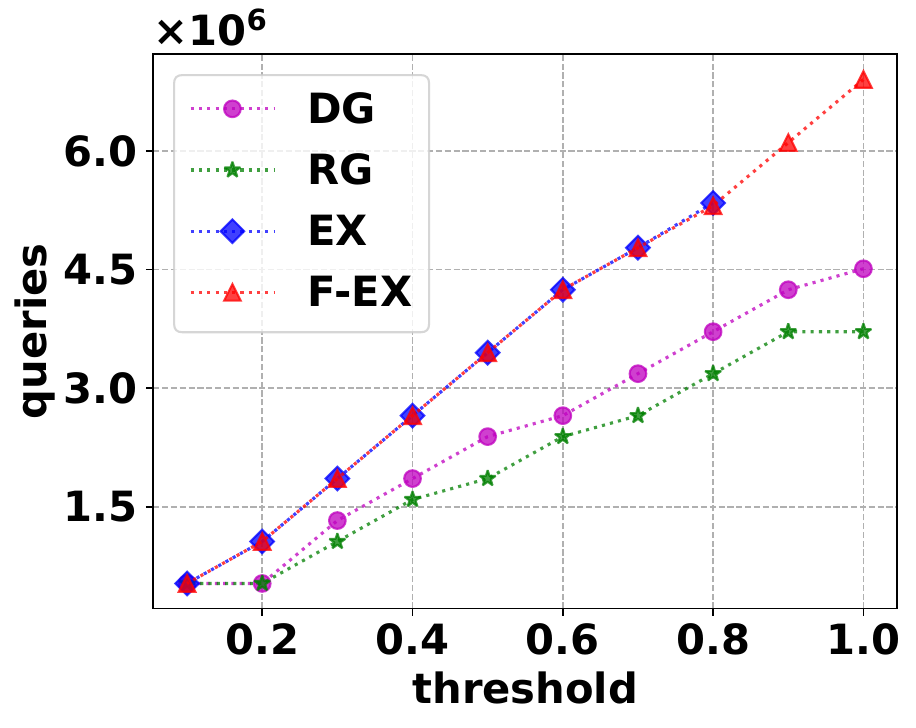}}
\hspace{-0.5em}
\subfigure[euall queries]
{\label{fig:euall_eps_q}\includegraphics[width=0.24\textwidth]{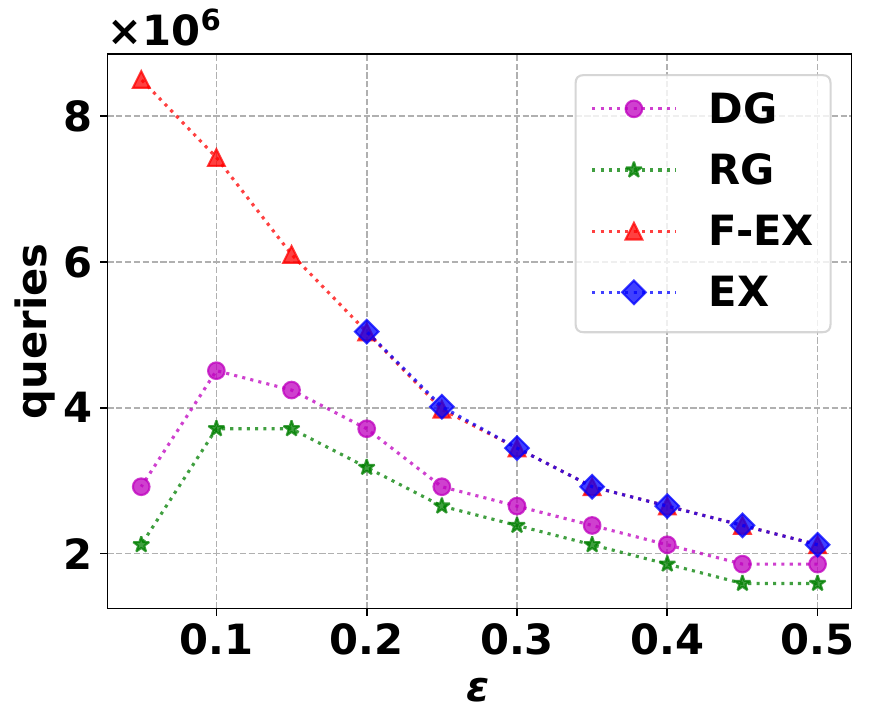}}
\caption{The experimental results of running the monotone algorithms on instances of data summarization on the delicious URL dataset ("cover") and running \filt on the instances of graph cut on the email-EuAll graph ("euall").}
        \label{fig:}
\end{figure*}

\subsection{Non-Monotone Submodular Objective}
\label{sec:experiments_nonmono}
We now analyze the performance of \filt on several instances of graph cut over real social network data. The universe $U$ is all nodes in the network, and $f$ is the number of edges between a set and its complement. The network featured in the main paper is the email-EuAll dataset ($n=265214$, 420045 edges) from the SNAP large network collection \citep{leskovec2016snap} and additional datasets can be found in the appendix.  We run \filt with input $\epsilon$ in the range $(0,0.5)$ and threshold values between 0 and $f(X)$ where $X$ is a solution returned by the unconstrained submodular maximization algorithm of \cite{buchbinder2015tight} on the instance. When $\epsilon$ is varied, $\tau$ is fixed at $0.9f(X)$. When $\tau$ is varied, $\epsilon$ is fixed at $0.15$.

We compare the performance of \filt using several possible algorithms for the subroutine of \smp over $\cup S_i$ (see line \ref{alg:exact_step} in Algorithm \ref{alg:filter}), including a polynomial time approximation algorithm and an unconstrained submodular maximization algorithm. In particular, we use the random greedy approximation algorithm for \smp that is proposed in \cite{buchbinder2014submodular} ("RG"), and the double greedy approximation algorithm for unconstrained submodular maximization proposed in \cite{buchbinder2015tight} ("DG"). Random greedy and double greedy are both approximation algorithms ($1/e$ in expectation and $1/2$ in expectation respectively), and therefore the stopping conditions are set to be $\frac{(1-\epsilon)\tau}{e}$ and $\frac{(1-\epsilon)\tau}{2}$ respectively. We also consider an exact algorithm ("EX"), which essentially is a greedy heuristic followed by an exact search of all (exponentially many) possible solutions if the greedy fails. On instances where the exact algorithm was unable to complete in a time period of $5$ minutes, we did not include a data point. We further discuss the use of these algorithms in the appendix.

Before introducing the fourth subroutine, we discuss an interesting pattern that we saw in our instances of graph cut. We noticed that it was often the case that: (i) $\cup S_i$ tended to be small compared to its upper bound and in fact typically $|\cup S_i|$ was smaller than the \smp constraint, making the subroutine an instance of unconstrained submodular maximization; (ii) The majority of elements (if not all) were "monotone" in the sense that for many $x\in\cup S_i$, $\Delta f(\cup S_i/x,x)\geq 0$. Let $M\subseteq\cup S_i$ be the set of monotone elements. Then if (i) holds the instance of submodular maximization is equivalent to $\arg\max_{X\in \cup S_i/M}f(X\cup M)$. If $M$ is large in $\cup S_i$, this new problem instance is relatively easy to solve exactly. This motivates our fourth algorithm, fast-exact ("F-EX"), used on instances where (i) holds, and is to separate $\cup S_i$ into monotone and non-monotone and search for the best subset amongst the non-monotone elements in a similar manner as the plain exact algorithm. We explore to what extent properties (i) and (ii) hold on different instances, as well as give additional details about the fast exact algorithm, in the appendix.

The results in terms of the $f$ values and size of the output solutions returned by the four algorithms are plotted in Figure \ref{fig:euall_tau_f} and Figure \ref{fig:euall_tau_c}. From the plots, one can see that the $f$ values satisfy that $f(S_{\text{exact}})\approx f(S_{\text{f-exact}})>f(S_{\text{DG}})>f(S_{\text{RG}})$. This is due to the stopping conditions for each algorithm, which follow from each algorithms approximation guarantee on $f$ of $1-\epsilon$, $1-\epsilon$, 1/2, and $1/e$ respectively. On the other hand, the size mirrors the $f$ value, since it tends to be the case that reaching a higher $f$ value requires more elements from $U$. The number of queries made by the algorithms can be seen in Figure \ref{fig:euall_eps_q} and \ref{fig:euall_tau_q}. As expected, the exact algorithms make more queries compared to the approximation algorithms, and in some cases "EX" doesn't even finish. However, by taking advantage of the properties (i) and (ii) discussed above, "F-EX" is able to run even for smaller $\epsilon$. Therefore, depending on the application, an exact algorithm on the relatively small set $\cup S_i$ may be a practical choice in order to achieve a solution that is very close to feasible.
\newpage
\bibliography{paper,icml2019}
\newpage
\section{Appendix for Section \ref{section:theoretical}}
\label{appdx:theo_sec}
In this portion of the appendix, we present missing details and proofs from Section \ref{section:theoretical} in the main paper. We first present missing content from Section \ref{section:monotone_objectives} in Section \ref{appdx:mono_sec}, followed by missing content from Section \ref{section:nonmono} in Section \ref{appendix:nonmono}, and finally missing content from Section \ref{section:reg} is presented in Section \ref{appendix:reg}.
In this section, we include a more detailed comparison of the algorithms presented in this paper.
In addition, pseudocode for the algorithm \conv of \cite{iyer2013} is presented in Algorithm \ref{alg:conver}.

\begin{table}[t!]
\centering
\caption{Theoretical guarantees of a subset of algorithms in this paper}
\label{tab:theo_results}
\begin{tabular}{|c|c|c|c|}
\hline
Alg name & $f$ value & soln size & number of queries \\
\hline
\gc & $(1-\epsilon)\tau$ & $\ln(1/\epsilon)$ & $O(n\ln(1/\epsilon)|OPT|)$\\
\hline
\sgc & $(1-\epsilon)\tau$ & $(1+\alpha)\ln(3/\epsilon)$ & $O(\frac{\alpha}{1+\alpha}n\ln(1/\delta)\ln^2(1/\epsilon)\ln(|OPT|))$\\
\hline
\tgc & $(1-\epsilon)\tau$ & $\ln(2/\epsilon)+1$ & $O(\frac{n}{\epsilon}\ln(\frac{|OPT|}{\epsilon}))$\\
\hline
\filt & $(1-\epsilon)\tau$ & $(1+\alpha)(2/\epsilon+1)|OPT|$ & $O(\log(|OPT|)(\frac{n}{\epsilon} + \mathcal{T}((1+\alpha)|OPT|/\epsilon^2))$\\

\hline

\end{tabular}
\end{table}

\begin{algorithm}[t!]
\caption{\conv}\label{alg:conver}
\textbf{Input}: An SC instance with threshold $\tau$, a $(\gamma,\beta)$-bicriteria approximation algorithm for \smp, $\alpha>0$\\
\textbf{Output}: $S\subseteq U$
\begin{algorithmic}[1]
\State $\kappa\gets(1+\alpha)$, $S\gets\emptyset$.
\While{$f(S)<\gamma\tau$}\label{line:convtest}
\State $S\gets(\gamma,\beta)$-bicriteria approximation for \smp with budget $\kappa$\label{line:convmax}
\State $\kappa\gets(1+\alpha)\kappa$
\EndWhile
\State \textbf{return} $S$
\end{algorithmic}
\end{algorithm}

\subsection{Additional content to Section \ref{section:monotone_objectives}}
\label{appdx:mono_sec}
In this section, we present the detailed statement and proofs of the randomized converting theorem \convr, and the proofs for the theoretical results of the \tgc and \sgc algorithms. The pseudocode for \tgc is in Algorithm \ref{alg:tg}. We now provide a proof of Theorem \ref{thm:threshold}.

\textbf{Theorem \ref{thm:threshold}. }\textit{
\tgc produces a solution with $(\ln(2/\epsilon)+1,1-\epsilon)$-bicriteria approximation guarantee to \mscp, in $O(\frac{n}{\epsilon}\log(\frac{n}{\epsilon}))$ number of queries of $f$.
}

\begin{proof}
   Let $r_0$ be defined as $r_0=\ln(\frac{2}{\epsilon})|OPT|$. By Claim 3.1 in \cite{badanidiyuru2014fast}, we have that any $s\in U$ that is added to $S$ on Line \ref{line:threshadd} of Algorithm \ref{alg:tg} would satisfy
   \begin{align*}
       \Delta f(S,s)\geq\frac{1-\epsilon/2}{|OPT|} \sum_{u\in OPT/S}\Delta f(S,u)
   \end{align*}
   when it was added.
   It follows by submodularity that $$\Delta f(S,s)\geq\frac{1-\epsilon/2}{|OPT|} \sum_{u\in OPT/S}\Delta f(S,u)\geq\frac{1-\epsilon/2}{|OPT|}(f(OPT)-f(S)).$$
   And by re-arranging the above equation, and using induction over the elements added to $S$, we have that
   $$f(S)\geq \left(1-\left(1-\frac{1-\epsilon/2}{|OPT|}\right)^i\right)f(OPT).$$
   Then at the $r_0$-th time of adding new element, it is the case that
   \begin{align*}
       f(S)\geq \left(1-\left(1-\frac{1-\epsilon/2}{|OPT|}\right)^{r_0}\right)f(OPT)\geq (1-\epsilon)f(OPT)\geq (1-\epsilon)\tau.
   \end{align*}
Thus the condition on Line \ref{line:threshstop} is satisfied by this point and the algorithm stops after adding the $r_0$-th element. The output solution set satisfies that $|S|\leq |r_0|$.

We now analyze the number of queries made to $f$. We claim that the algorithm stops before $w=\epsilon\max_{u\in U}f(u)/|OPT|$. This is because at the end of the iteration of the \textbf{while} loop in Algorithm \ref{alg:tg} corresponding to $w=\epsilon\max_{u\in U}f(u)/|OPT|$, we have that $\Delta f(S,s)< w$ for any $s\in U\setminus S$. Therefore
\begin{align*}
    \sum_{s\in OPT/S}\Delta f(S,s)< |OPT/S|w\leq \epsilon\max_{u\in U}f(u)\leq \epsilon f(OPT).
\end{align*}
It follows that at the end of this iteration of the \textbf{while} loop that $f(S)> (1-\epsilon)f(OPT)$, which means the condition on Line \ref{line:threshstop} is not satisfied and the algorithm terminates. The number of iterations is thus $O\left(\frac{\ln(|OPT|/\epsilon)}{\epsilon}\right)$ and the query complexity is $O\left(\frac{n}{\epsilon}\ln\left(\frac{|OPT|}{\epsilon}\right)\right)$.
\end{proof}

\begin{algorithm}[t]
\caption{\tgc}\label{alg:tg}
\textbf{Input}: $\epsilon$\\
\textbf{Output}: $S\subseteq U$
\begin{algorithmic}[1]
  \State $S\gets\emptyset$
  \State $w=\max_{u\in U}f(u)$
  \While{$f(S)<(1-\epsilon)\tau$}\label{line:threshstop}
      \For { each $u$ in $U$}
      \If {$\Delta f(S,u)\geq w $}
        \State $S\gets S\cup \{u\}$\label{line:threshadd}
      \EndIf
      \EndFor
  \State $w=w(1-\frac{\epsilon}{2})$
  \EndWhile
  \State 
 \textbf{return} $S$
\end{algorithmic}
\end{algorithm}

Next, we consider the algorithm \convr that converts algorithms for monotone \smp to ones for \mscp. Pseudocode for \convr is provided in Algorithm \ref{alg:conv_rand}. We now present the proof of Theorem \ref{thm:convert}.

\begin{algorithm}[t]
\caption{\convr}\label{alg:conv_rand}
\textbf{Input}: A \mscp instance with threshold $\tau$, a $(1-\frac{\epsilon}{2},\beta)$-bicriteria approximation algorithm for monotone \smp where $\epsilon$ is in expectation, $\alpha>0$\\
\textbf{Output}: $S\subseteq U$
\begin{algorithmic}[1]
\State $S_i\gets\emptyset$, $\forall i\in\{1,..., \ln(1/\delta)/\ln(2)\}$
\State $g\gets(1+\alpha)$
\While{$f(S_i)<(1-\epsilon))\tau$ $\forall i$}\label{line:convrtest}
    \For{$i\in\{1,..., \ln(1/\delta)/\ln(2)\}$}
        \State $S_i\gets(1-\epsilon/2,\beta)$-bicriteria approximation for monotone \smp with objective function $f_\tau$ and budget $g$\label{line:convrmax}
    \EndFor
\State $g\gets(1+\alpha)g$
\EndWhile
\State \textbf{return} $S$
\end{algorithmic}
\end{algorithm}

\textbf{Theorem \ref{thm:convert}. }\textit{
   Any randomized $(1-\epsilon/2,\beta)$-bicriteria approximation algorithm for monotone \smp that runs in time $\mathcal{T}(n)$ where $\epsilon$ holds only in expectation can be converted into a $((1+\alpha)\beta,1-\epsilon)$-bicriteria approximation algorithm for \mscp that runs in time $O(\log_{1+\alpha}(|OPT|)\ln(1/\delta)\mathcal{T}(n))$ where $\epsilon$ holds with probability at least $1-\delta$.}

\begin{proof}
Consider the run of the algorithm for monotone \smp on Line \ref{line:convmax} of Algorithm \ref{alg:conv_rand} when the parameter $g$ falls into the region $|OPT|\leq g\leq(1+\alpha)|OPT|$. Notice that it is the case that $f_\tau$ is also monotone and submodular. By the theoretical guarantees of the algorithm for monotone \smp, we have that for all $i\in\{1,..., \ln(1/\delta)/\ln(2)\}$
\begin{align*}
    \mE f_\tau(S_i)&\geq(1-\epsilon/2)\max_{|X|\leq g}\{f_\tau(X)\}\geq(1-\epsilon/2)\max_{|X|\leq |OPT|}\{f_\tau(X)\}\geq(1-\epsilon/2)\tau.
\end{align*}

From Markov's inequality, we have that for each $i\in\{1,..., \ln(1/\delta)/\ln(2)\}$
\begin{align*}
    P(f_\tau(S_i)\leq(1-\epsilon)\tau)\leq P(\tau-f_{\tau}(S_i)>2(\tau-\mathbb{E}f_\tau(S_i)))
    \leq1/2.
\end{align*}
Then the probability that none of the subsets $S_i$ can reach the stopping condition can be bounded by 
\begin{align*}
    P(f(S_i)\leq(1-\epsilon)\tau, \forall i)&=P(f_\tau(S_i)\leq(1-\epsilon)\tau, \forall i)\notag\\
&=\prod_{i=1}^{\ln(1/\delta)/\ln(2)}P(f_\tau(S_i)\leq(1-\epsilon)\tau)\\
&\leq(1/2)^{\ln(1/\delta)/\ln(2)}=\delta.
\end{align*}
This means with probability at least $1-\delta$, \convr stops when $g$ reaches the region where $|OPT|\leq g\leq(1+\alpha)|OPT|$ since the condition of the \textbf{while} loop is not satisfied. 

\end{proof}

Before we present the proof of the theoretical guarantees of the algorithm \sgc, we first introduce an algorithm for monotone \smp with bicriteria approximation guarantee called \sg. \sg is an extension of the stochastic greedy algorithm of \cite{mirzasoleiman2015lazier}, which produces an infeasible solution to the instance of monotone \smp that has $f$ value arbitrarily close to that of the optimal solution. Pseudocode for \sg is provided in Algorithm \ref{alg:sg}. Notice that if we use the \sg as a subroutine for \convr, we can obtain an algorithm for \mscp that runs in $O(n\ln^2(3/\epsilon)\ln(1/\delta)\log_{1+\alpha}|OPT|)$, which is less queries than the standard greedy algorithm, but worse than our algorithm \sgc by a factor of $(1+\alpha)/\alpha$.

\begin{algorithm}[t]
\caption{\sg}
\textbf{Input: }$\epsilon > 0$\\
\textbf{Output: }$S\subseteq U$
\label{alg:sg}
\begin{algorithmic}[1]

  \State $S\gets\emptyset$
  \While{$|S|< \ln(\frac{3}{2\epsilon})\kappa$}
    \State $R\gets$ sample $\frac{n}{\kappa}\ln(\frac{3}{2\epsilon})$ elements from $U$
    \State $u\gets\argmax_{x\in R}\Delta f(S,x)$
  \EndWhile
  \State \textbf{return} $S$
\end{algorithmic}
\end{algorithm}

\noindent
\begin{theorem}
    \label{thm:stochbi}
    \sg produces a solution with $(1-\epsilon,\ln(\frac{3}{2\epsilon}))$-bicriteria approximation guarantee to \msmp, where the $1-\epsilon$ holds in expectation, in $O(n\ln^2(\frac{3}{2\epsilon}))$ queries of $f$. 
\end{theorem}
\begin{proof}
Our argument to prove Theorem \ref{thm:stochbi} follows a similar approach as \cite{mirzasoleiman2015lazier}.
Let the partial solution $S$ before $i$-th iteration of the \textbf{while} loop of Algorithm \ref{alg:sg} be $S_i$, the item that is added to set $S$ during iteration $i$ be $u_i$, and the sampled set $R$ during iteration $i$ be $R_i$. Consider the beginning of any iteration $i$ of the \textbf{while} loop.
Then because of the greedy choice for $u_i$ it is the case that if $R_i\cap OPT\neq\emptyset$,
\begin{align*}
    \Delta f(S_i,u_i)\geq \Delta f(S_i,x)
\end{align*}
for all $x\in R_i\cap OPT$. Therefore we have that in expectation over the randomly chosen set $R_i$
\begin{align}
\label{eqn:stoc_greedy}
    \mE[\Delta f(S_i,u_i)|S_i] 
    &\geq\textbf{Pr}(R_i\cap OPT\neq \emptyset)\mE[\Delta f(S_i,u_i)|S_i, R_i\cap OPT\neq \emptyset]\notag \\
    &\geq\textbf{Pr}(R_i\cap OPT\neq \emptyset)\frac{1}{|R_i\cap OPT|}\sum_{x\in R_i\cap OPT} \Delta f(S_i,x) \notag \\
    &= \textbf{Pr}(R_i\cap OPT\neq \emptyset)\frac{1}{\kappa}\sum_{x\in OPT} \Delta f(S_i,x)
\end{align}
where the last inequality follows since the elements of $R_i$ are chosen uniformly randomly, implying that all elements of $OPT$ are equally likely to be chosen.

Let the number of samples $s=\frac{n}{k}\ln(\frac{3}{2\epsilon})$.
The probability $\textbf{Pr}(R_i\cap OPT\neq \emptyset)$ can be lower bounded as
\begin{align}
    \label{eqn:prob_of_intersection}
    \textbf{Pr}(R_i\cap OPT\neq \emptyset)&=1-\textbf{Pr}(R_i\cap OPT= \emptyset)\notag\\
    &=1-\frac{\binom{n-\kappa}{s }}{\binom{n}{s}}\notag\\
    &\geq1-(\frac{n-\kappa}{n})^{s }\notag\\
    &\geq1-\frac{2\epsilon}{3} ,
\end{align}
where the last inequality comes from the fact that $(1-\frac{1}{x})^x\leq e^{-1}$. 
By submodularity, we have 
\begin{align}
\label{eqn:submodularity_CSM}
   \sum_{u\in OPT} \Delta f(S_i,u)\geq f(OPT)-f(S_i).
\end{align}
Plug Equations (\ref{eqn:prob_of_intersection}) and (\ref{eqn:submodularity_CSM}) into (\ref{eqn:stoc_greedy}) yields that
\begin{align}
    \mE[\Delta f(S_i,u_i)|S_i]\geq \frac{1-2\epsilon/3 }{\kappa }(f(OPT)-f(S_i)).
\end{align}
Therefore,
\begin{align*}
    \mE[ f(S_{i+1})|S_i]\geq \frac{1-2\epsilon/3 }{\kappa }f(OPT)+(1-\frac{1-2\epsilon/3 }{\kappa })f(S_i)).
\end{align*}
By induction and by taking expectation over $S_i$, we can get that at the last iteration
\begin{align*}
    \mE[ f(S)]&\geq \frac{1-2\epsilon/3 }{\kappa }\sum_{i=0}^{|S|-1}(1-\frac{1-2\epsilon/3 }{\kappa })^if(OPT)\\
    &\geq 
    \{1-(1-\frac{1-2\epsilon/3 }{\kappa })^{|S|}\}f(OPT)\\
    &\geq 
    (1-(2\epsilon/3)^{1-2\epsilon/3})f(OPT)\geq  
    (1-\epsilon)f(OPT),
\end{align*}
where in the last inequality, we use the fact that $|S|=\ln(3/2\epsilon)\kappa$.
\end{proof}

We now present the omitted proofs for Lemmas \ref{lemma:marggain}, \ref{lem:sgc_expection}, and \ref{lem:sgc_high_prob} about our algorithm \sgc. Once these lemmas are proven, we next prove Theorem \ref{theorem:stochastic} using the lemmas.

\textbf{Lemma \ref{lemma:marggain}. }\textit{
Define $\ell$ to be the integer such that $(1+\alpha)^{\ell-1}\leq |OPT| \leq (1+\alpha)^{\ell}$.
Consider any of the sets $S_i$ at the beginning of an iteration on Line \ref{sgc:iter_start} where $g\leq (1+\alpha)^{\ell}$.
    Then if $u_i$ is the random element that will be added on Line \ref{sgc:add_new}, we have that  $$\mathbb{E}[\Delta f_{\tau}(S_i, u_i)|S_i]\geq \frac{1-\epsilon/3}{(1+\alpha)^{\ell}}(\tau - f_{\tau}(S_i))  \geq \frac{1-\epsilon/3}{(1+\alpha)|OPT|}(\tau - f_{\tau}(S_i)).$$}
\begin{proof}
    Define $OPT^*=\arg\max_{|X|\leq(1+\alpha)^{\ell}}f(X).$ Then we have $f(OPT^*)\geq f(OPT)\geq\tau$. Therefore $f_{\tau}(OPT^*)\geq\tau$.
    Consider an iteration of the \textbf{while} loop on Line \ref{sgc:iter_start} of \sgc where $g\leq (1+\alpha)^{\ell}$, and some iteration of inside \textbf{for} loop corresponding to set $S_i$. 
    From Line $\ref{line:sample_R}$ of \sgc, we have that the size of the randomly sampled subset $R_i$ for $S_i$ is $\min\{n, n\ln(3/\epsilon)/g\}$. Then if $n\ln(3/\epsilon)/g<n$, we have that 
    \begin{align*}
        |R_i|= \frac{n}{g}\ln(3/\epsilon)\geq\frac{n}{(1+\alpha)^{\ell}}\ln(3/\epsilon).
    \end{align*}
    This implies that if $n\ln(3/\epsilon)/g<n$, when we randomly sample $R_i$
    \begin{align*}
        P(R_i\cap OPT^* = \emptyset) = (1-|OPT^*|/n)^{|R_i|}\leq (1-(1+\alpha)^{\ell}/n)^{\frac{n\ln(3/\epsilon)}{(1+\alpha)^{\ell}}}
        \leq \epsilon/3.
    \end{align*}
    On the other hand, if $n\ln(3/\epsilon)/g \geq n$, then $P(R_i\cap OPT^* = \emptyset)=0$ and the above inequality also holds.

    Notice that $f_{\tau}$ is monotone and submodular.
    Let $u_i$ be the elements of $R_i$ which will be added to $S_i$. Then
    \begin{align*}
        \mathbb{E}[\Delta f_\tau(S_i, u_i)| R_i\cap OPT^* \neq\emptyset]&\geq \frac{1}{|OPT^*|}\sum_{o\in OPT^*}\Delta f_\tau(S_i,o)\\
         &\explaineq{(i)}{\geq}\frac{1}{|OPT^*|}(f_\tau(OPT^*) - f_\tau(S_i))\\
        &\geq \frac{1}{(1+\alpha)^{\ell}}(\tau - f_\tau(S_i)).
    \end{align*}
    where (i) is implied by the fact that $f_{\tau}$ is monotone and submodular.
    Altogether we have that 
    \begin{align*}
        \mathbb{E}[\Delta f_\tau(S_i, u_i)|S_i]&\geq P(R\cap OPT^* \neq \emptyset) \mathbb{E}[\Delta f_\tau(S_i, u_i)| R\cap OPT^* \neq\emptyset]\\
        &\geq\frac{1-\epsilon/3}{(1+\alpha)^{\ell}}(\tau - f_{\tau}(S_i)).
    \end{align*}
\end{proof}

\noindent
\textbf{Lemma \ref{lem:sgc_expection}. }\textit{Once $r\geq (1+\alpha)\ln(3/\epsilon)|OPT|$, we have that $\mathbb{E}[f_{\tau}(S_i)] \geq\left(1-\frac{\epsilon}{2}\right)\tau$ for all $i$.}
\begin{proof}
    Define $\ell$ to be the integer such that $(1+\alpha)^{\ell-1}\leq |OPT| \leq (1+\alpha)^{\ell}$. Then notice that throughout \sgc until $r$ is incremented to be $\ln(3/\epsilon)(1+\alpha)^{\ell}$, it is the case that $g\leq (1+\alpha)^{\ell}$. Therefore Lemma \ref{lemma:marggain} applies for all marginal gains of adding elements to the sets until the end of the $\ln(3/\epsilon)(1+\alpha)^{\ell}$ iteration of the \textbf{while} loop.
    Consider any of the sets $S_i$.
    By applying recursion to Lemma \ref{lemma:marggain}, we have that at the end of any iteration $r\leq \ln(3/\epsilon)(1+\alpha)^{\ell}$ of the \textbf{while} loop
    \begin{align*}
        \mathbb{E}[ f_\tau(S_i)]\geq\left(1-\left(\frac{1-\epsilon/3}{(1+\alpha)^{\ell}}\right)^r\right)\tau.
    \end{align*}
    Therefore, Once $r$ reaches $\ln(3/\epsilon)(1+\alpha)^{\ell}$, the expectation of $S_i$ is
    \begin{align*}
         \mathbb{E}[ f_\tau(S_i)]\geq\left(1-\left(\frac{1-\epsilon/3}{(1+\alpha)^{\ell}}\right)^{(1+\alpha)^{\ell}\ln(3/\epsilon)}\right)\tau\geq\left(1-\frac{\epsilon}{2}\right)\tau.
    \end{align*}
    Because $\ln(3/\epsilon)(1+\alpha)^{\ell}\leq (1+\alpha)\ln(3/\epsilon)|OPT|$ and $f$ is monotonic, Lemma \ref{lem:sgc_expection} holds.
\end{proof}

\noindent
\textbf{Lemma \ref{lem:sgc_high_prob}. }\textit{With probability at least $1-\delta$, once $r\geq (1+\alpha)\ln(3/\epsilon)|OPT|$, we have that $\max_{i}f(S_i)\geq(1-\epsilon)\tau$.}



\begin{proof}
    Since $f_\tau(S)=\min\{f(S),\tau\}$, we have that $\tau-f_\tau(S)\in[0,\tau]$ for any $S\in  U$. By applying the Markov's inequality, we have that
    \begin{align*}
        P(\tau-f_\tau(S_j)\geq 2\mathbb{E}(\tau-f_\tau(S_j))\leq\frac{1}{2}
    \end{align*}
From Lemma \ref{lem:sgc_expection}, we have $\mathbb{E}(\tau-f_\tau(S_j))\leq\frac{\epsilon\tau}{2}$. Then 
\begin{align*}
    P(f_\tau(S_j)\leq(1-\epsilon)\tau)&=P(\tau-f_\tau(S_j)\geq\epsilon\tau)\\
    &\leq P(\tau-f_\tau(S_j)\geq 2\mathbb{E}(\tau-f_\tau(S_j)))\\
    &\leq\frac{1}{2}.
\end{align*}
Thus we have
\begin{align*}
    P(\max_j f_\tau(S_j)>(1-\epsilon)\tau)&=1-P( f_\tau(S_j)\leq(1-\epsilon)\tau,\forall j)\\
    &=1- P( f_\tau(S_1)\leq(1-\epsilon)\tau)^{\ln(1/\delta)/\ln(2)}\\
    &\geq1-\left(\frac{1}{2}\right)^{\ln(1/\delta)/\ln(2)}\geq1-\delta.
\end{align*}
\end{proof}

\textbf{Theorem \ref{theorem:stochastic}. }\textit{   Suppose that \sgc is run for an instance of \mscp.
   Then with probability at least $1-\delta$, \sgc
   outputs a solution $S$ that satisfies a $((1+\alpha)\ln(3/\epsilon),1-\epsilon)$-bicriteria approximation guarantee in at most $$O\left(\frac{\alpha}{1+\alpha}n\ln(1/\delta)\ln^2(3/\epsilon)\log_{1+\alpha}(|OPT|)\right)$$ queries of $f$.}
    \begin{proof}
        From Lemma \ref{lem:sgc_high_prob}, we have that the algorithm stops by the time $r$ reaches $(1+\alpha)\ln(3/\epsilon)|OPT|$ with probability at least $1-\delta$. Before $r$ reaches this point, consider the number of queries made to $f$ over the duration that $g$ is a certain value $(1+\alpha)^{m}$. Then for each of the $O(\ln(1/\delta))$ sets a total of at most
        \begin{align*}
            \left(\ln\left(\frac{3}{\epsilon}\right)(1+\alpha)^m-\ln\left(\frac{3}{\epsilon}\right)(1+\alpha)^{m-1}\right)\frac{n\ln(\frac{3}{\epsilon})}{(1+\alpha)^m}
            = \frac{\alpha}{1+\alpha}n\ln^2(3/\epsilon)
        \end{align*}
        queries are made to $f$. Because $g$ takes on at most $O(\log_{1+\alpha}(|OPT|))$ values before $r$ reaches $(1+\alpha)\ln(3/\epsilon)|OPT|$, Theorem \ref{theorem:stochastic} follows.
    \end{proof}

\subsection{Additional content for Section \ref{section:nonmono}}
\label{appendix:nonmono}
In this section, we present the proofs of theoretical results from Section \ref{section:nonmono}. First, in Claim \ref{claim:mono} we present a useful claim from \cite{crawford2023scalable} about submodular functions in order to prove Lemma \ref{lemma:filt}. Next, we use Claim \ref{claim:mono} in order to prove Lemma \ref{lemma:filt}. Finally, Lemma \ref{lemma:filt} is used to prove the main result for \filt, Theorem \ref{theorem:filt}.

\begin{claim}[Claim 1 in \cite{crawford2023scalable}]
\label{claim:mono}
Let $A_1,...,A_m\subseteq U$ be disjoint, and $B\subseteq U$.
Then there exists $i\in\{1,...,m\}$ such that
$f(A_i\cup B)\geq (1-1/m)f(B)$.
\end{claim}

\textbf{Lemma \ref{lemma:filt}. }
\textit{ By the time that $g$ reaches the region $[|OPT|, (1+\alpha)|OPT|]$ and the loop on Line \ref{line:loop} of \filt has completed, there exists a set $X\subseteq \cup S_i$ of size at most $(2/\epsilon+1)g$ such that
    $f(X)\geq (1-\epsilon)\tau$.}
\begin{proof}
    Suppose that \filt has reach the end of the loop on Line \ref{line:loop} when $g\geq |OPT|$. We first consider the case where there exists some set $S_t$ $t\in\{1,...,2/\epsilon\}$ such that $|S_t|=2g/\epsilon$. In this case it follows that $f(S_t)=\sum_{i=1}^{2g/\epsilon}\Delta f(S_t^i,u_i)\geq \tau$, where $S_t^i$ is the set of $S_t$ before adding the $i$-th element $u_i$, and so the Lemma statement is proven.

    Next, we consider the case where $|S_j|<2g/\epsilon$, $\forall j\in[2/\epsilon]$. Let $OPT_1=OPT\cap (\cup_{i=1}^{2/\epsilon} S_i)$ and $OPT_2=OPT/OPT_1$. By Claim \ref{claim:mono}, there exists a set $S_t$ such that
    $$f(S_t\cup OPT)\geq (1-\epsilon/2)f(OPT)\geq(1-\epsilon/2)\tau.$$
    Since $|S_t|<2g/\epsilon$ at the end of the algorithm, we can see each element $o$ in $OPT_2$ is not added into $S_t$ because $\Delta f(S_t,o)\le\epsilon \tau/(2 g)$ at the time $o$ is seen in the loop on Line \ref{line:loop}. By submodularity and the fact that $g\geq |OPT|\geq |OPT_2|$, 
    $$\sum_{o\in OPT_2}\Delta f(S_t\cup OPT_1,o) < \epsilon|OPT_2|\tau/(2g) \leq \epsilon \tau/2.$$
    Therefore
    \begin{align*}
        (1-\epsilon/2)\tau &\leq f(S_t \cup OPT) \\
        &\leq f(S_t\cup OPT_1) + \Delta f(S_t\cup OPT_1, OPT_2)\\
        &\leq f(S_t\cup OPT_1) + \sum_{o\in OPT_2}\Delta f(S_t\cup OPT_1, o) \\
        &\leq f(S_t\cup OPT_1) + \epsilon \tau/2.
    \end{align*}
    Therefore $f(S_t\cup OPT_1)\geq (1-\epsilon)\tau$.
    Because $S_t\cup OPT_1\subseteq \cup S_i$, and it is the case that $|S_t\cup OPT_1|\leq |S_t|+|OPT_1|\leq 2g/\epsilon + g$,
    so the Lemma statement is proven.
\end{proof}

\textbf{Theorem \ref{theorem:filt}. }
   \textit{Suppose that \filt is run for an instance of \scp. Then \filt returns $S$ such that $f(S)\geq (1-\epsilon)\tau$ and $|S|\leq (1+\alpha)(2/\epsilon+1)|OPT|$ in at most
  $$\log_{1+\alpha}(|OPT|)
    \left(\frac{2n}{\epsilon}+
    \mathcal{T}\left((1+\alpha)\left(\frac{4}{\epsilon^2}|OPT|\right)\right)\right)$$
    queries of $f$, where $\mathcal{T}(m)$ is the number of queries to $f$ of the algorithm for \smp used on Line \ref{line:max} of Algorithm \ref{alg:filter} on an input set of size $m$.}
\begin{proof}
    By Lemma \ref{lemma:filt}, once \filt reaches $g\geq |OPT|$ then $f(S)\geq(1-\epsilon)\tau$ will be satisfied and therefore the \textbf{while} loop will exit. This implies that $g\leq(1+\alpha)|OPT|$ once \filt exits, and therefore $|S|\leq (2/\epsilon+1)g\leq (1+\alpha)(2/\epsilon+1)|OPT|$. Therefore the qualities of $S$ stated in Theorem \ref{theorem:filt} are proven. As far as the number of queries of $f$, it takes at most $\log_{1+\alpha}(|OPT|)$ iterations of the loop to reach the point that $g\geq |OPT|$. In addition, each iteration of the loop makes $\frac{2n}{\epsilon}+
    \mathcal{T}\left((1+\alpha)\left(\frac{4}{\epsilon^2}|OPT|\right)\right)$ queries of $f$. Therefore the bound on the number of queries of $f$ in Theorem \ref{theorem:filt} is proven.
\end{proof}

\subsection{Supplementary material to Section \ref{section:reg}}
\label{appendix:reg}
We now present additional theoretical details for Section \ref{section:reg}, where we considered regularized \scp. First, we prove Theorem \ref{thm:rconv} about our algorithm \regconv which converts algorithms for regularized \smp into ones for \rscp.

\textbf{Theorem \ref{thm:rconv}. }\textit{Suppose that we have an algorithm \reg for maximization of a regularized submodular function subject to a cardinality constraint $\kappa$, and that algorithm is guaranteed to return a set $S$ of cardinality at most $\rho\kappa$ such that
    $$g(S)-c(S)\geq \gamma g(X)-\beta c(X)$$ for all $X$ such that $|X|\leq\kappa$ in time $T(n)$.
    Then the algorithm \regconv using \reg as a subroutine returns a set $S$ such that
    $$g(S)-\frac{\gamma}{\beta}c(S)\geq \gamma\tau$$ and
    $|S|\leq (1+\alpha)\rho|OPT|$ in time $O(\log_{1+\alpha}(n)T(n))$.}
\begin{proof}
    Let $OPT$ be the optimal solution to the instance of regularized \scp.
    Consider the iteration of \regconv where $\kappa$ has just increased above $|OPT|$, i..e $|OPT|\leq\kappa\leq (1+\alpha)|OPT|$. Then we run \reg with input objective $g-\frac{\gamma}{\beta}c$ and budget $\kappa\geq|OPT|$. Then by the assumptions on \reg we have that
    $$g(S)-\frac{\gamma}{\beta}c(S)\geq \gamma(g(OPT)-c(OPT)) \geq\gamma\tau$$
    and $|S|\leq \rho\kappa\leq (1+\alpha)\rho|OPT|$.
\end{proof}

We now fill in the missing information concerning our algorithm \distgreed. Recall the definition of $$\Phi_i(X) = \left(1-\frac{1}{\kappa}\right)^{t-i}g(X) - c(X),$$ where $t=\ln(1/\epsilon)\kappa$, which is used in both the pseudocode for \distgreed and throughout the proofs. First, pseudocode for \distgreed is presented in Algorithm \ref{alg:distortedgreedy}. Next, we present and prove Lemma \ref{distlemma}, and then use Lemma \ref{distlemma} in order to prove our main result for \distgreed, Theorem \ref{distthm}.

\begin{lemma}
  \label{distlemma}
  Consider any iteration $i+1$ of the \textbf{while} loop in \distgreed. Let $S_i$ be defined to be $S$ after the $i$-th iteration of the \textbf{while} loop in \distgreed. Then for any $X\subseteq U$ such that $|X|\leq\kappa$, 
  \begin{align*}
    \Phi_{i+1}(S_{i+1})-\Phi_{i}(S_i) \geq
    \frac{1}{\kappa}\left(1-\frac{1}{\kappa}\right)^{t-(i+1)}g(X)-\frac{c(X)}{\kappa}.
  \end{align*}
\end{lemma}
\begin{proof}
  First, suppose that during iteration $i+1$ an element $s_{i+1}$ was added to $S$.
  First, notice that
  \begin{align}
  \label{eqn:phi_recursion1}
    \Phi_{i+1}(S_{i+1})-\Phi_i(S_i) & =\Phi_{i+1}(S_{i+1})-\Phi_{i+1}(S_i)+\Phi_{i+1}(S_{i})-\Phi_i(S_i)\notag\\
    &= \Delta\Phi_{i+1}(S_i,s_{i+1})
       + \frac{1}{\kappa}\left(1-\frac{1}{\kappa}\right)^{t-(i+1)}g(S_i).
  \end{align}
  In addition, notice that for any subset of $U$ such that $|X|\leq \kappa$
  \begin{align*}
     \Delta \Phi_{i+1} (S_i,s_{i+1}) &\overset{(a)}{\geq} \frac{1}{\kappa}\sum_{o\in X}\left(\Delta \Phi_{i+1} (S_i,o) \right) \\
    &\geq \frac{1}{\kappa}\sum_{o\in X}\left(\left(1-\frac{1}{\kappa}\right)^{t-(i+1)}\Delta g(S_i,o)-\Delta c(S_i,o)\right) \\
    &\geq \frac{1}{\kappa}\left(1-\frac{1}{\kappa}\right)^{t-(i+1)}\sum_{o\in X}\Delta g(S_i,o)-\frac{c(X)}{\kappa}\\
    &\overset{(b)}{\geq}\frac{1}{\kappa}\left(1-\frac{1}{\kappa}\right)^{t-(i+1)}\left(g(S_i\cup X)-g(S_i)\right)-\frac{c(X)}{\kappa}\\
    &\overset{(c)}{\geq}\frac{1}{\kappa}\left(1-\frac{1}{\kappa}\right)^{t-(i+1)}\left(g(X)-g(S_i)\right)-\frac{c(X)}{\kappa}
  \end{align*}
  where (a) is because of the greedy choice of $s_i$; (b) is by the submodularity of
  $g$; and (c) is by the monotonicity of $g$. Combining the previous two equations gives the result of Lemma \ref{distlemma}. 
  
  Next, we consider the case where $\Delta\Phi_i(S_i,x)\leq 0$ for all $x\in U$, and so there is no new element added. In this case, we can randomly choose an element from the current solution $S_i$ as the added element $s_{i+1}$. It is not hard to verify that the above two inequalities still holds in this case, and therefore again we have the result of Lemma \ref{distlemma}.
\end{proof}
\begin{algorithm}[t]
\caption{\distgreed}
\textbf{Input}: $\epsilon$\\
\textbf{Output}: $S\subseteq U$
\label{alg:distortedgreedy}
\begin{algorithmic}[1]
  \State $S\gets\emptyset$
  \State $i\gets 1$
  \While{$|S|<\ln(1/\epsilon)\kappa$}
    \State $u\gets\argmax_{x\in U}\Delta \Phi_i(S,x)$
    \If{$\Delta \Phi_i(S,x) > 0$}
      \State $S\gets S\cup\{u\}$
      \State $i\gets i+1$
    \Else
      \State \textbf{break}
    \EndIf
  \EndWhile
  \State \textbf{return} $S$
\end{algorithmic}
\end{algorithm}


\textbf{Theorem \ref{distthm}. }\textit{
  Suppose that \distgreed is run for an instance of regularized \smp. Then
  \distgreed produces a solution $S$ in $O(n\kappa\ln(1/\epsilon))$ queries of $f$ such that $|S|\leq\ln(1/\epsilon)\kappa$ and for all $X\subseteq U$ such that $|X|\leq\kappa$,
  $g(S)-c(S)\geq (1-\epsilon)g(X)-\ln(1/\epsilon)c(X).$
}
\begin{proof}
  We will use Lemma \ref{distlemma} to prove Theorem \ref{distthm}. Define $t=\ln(1/\epsilon)\kappa$. Then we see that
  \begin{align*}
    g(S_t)- c(S_t) &\overset{(a)}{\geq} \Phi_t(S_t) - \Phi_0(S_0) \\
    &= \sum_{i=0}^{t-1}\left(\Phi_{i+1}(S_{i+1})-\Phi_i(S_i)\right) \\
    &\overset{(b)}{\geq} \sum_{i=0}^{t-1}\left(\frac{1}{\kappa}\left(1-\frac{1}{\kappa}\right)^{t-(i+1)}g(X)-\frac{c(X)}{\kappa}\right)\\
    &= \frac{g(X)}{\kappa}\sum_{i=0}^{t-1}\left(1-\frac{1}{\kappa}\right)^{t-(i+1)}-\ln\left(1/\epsilon\right)c(X)\\
    &= \frac{g(X)}{\kappa}\sum_{i=0}^{t-1}\left(1-\frac{1}{\kappa}\right)^i-\ln\left(1/\epsilon\right)c(X)\\
    &\overset{(c)}{=} \left(1-\left(1-\frac{1}{\kappa}\right)^t\right)g(X)-\ln\left(1/\epsilon\right)c(X)\\
    &\geq(1-\epsilon)g(X)-\ln\left(1/\epsilon\right)c(X)
  \end{align*}
  where (a) is because $g(\emptyset)\geq 0$;
  (b) is because Lemma \ref{distlemma};
  and (c) is by the formula for geometric series.
\end{proof}

\section{Supplementary material for Section \ref{section:exp}}
In this section, we present supplementary material to Section \ref{section:exp}. In particular, we present additional detail about the experimental setup in Section \ref{appendix:setup}, and additional experimental results in Section \ref{appendix:results}.

\subsection{Experimental setup}
\label{appendix:setup}
First of all, we provide more detail about the two applications used to evaluate the algorithms proposed in the main paper. For monotone \scp, the application considered here is data summarization, where $f$ is a function that represents how well a subset could summarize the whole dataset. The problem definition is as follows.
\begin{definition}{(\textbf{Data Summarization})}
Suppose there are a total of $n$ elements denoted as $U$. Let $T$ be a set of tags. Each element in $U$ is tagged with a set of elements from $T$ via function $t: U\rightarrow 2^T$. The function $f$ is defined as
\begin{align*}
    f(S)=|\cup_{s\in S}t(s)|,\qquad\forall S\in U.
\end{align*}
\end{definition}
From the definition, we can see that $f$ is both monotone and submodular. For general \scp, the application we consider is where $f$ is a graph cut function, which is a submodular but not necessarily monotone function. 
\begin{definition}{(\textbf{Graph cut})}
Let $G = (V, E)$ be a graph, and $w:E\rightarrow \mathbb{R}_{\geq 0}$ be a function that assigns a weight for evry edge in the graph. The function $f:2^U\rightarrow \mathbb{R}_{\geq 0}$ maps a subset of vertices $X\subseteq V$ to the total weight of edges between $X$ and $V\backslash X$. More specifically,
\begin{align*}
    f(S)=\sum_{x\in X, y\in V\backslash X}w(x,y).
\end{align*}
\end{definition}

In addition to the datasets considered in the main paper,
we also look at graph cut instances on the com-Amazon ($n=334863$, 925872 edges), ego-Facebook ($n=4039$,	88234 edges), and email-Enron ($n=36692$, 183831 edges) graphs from the SNAP large network collection \citep{leskovec2016snap}. 

We now present more detail about the algorithms ``EX'' and ``F-EX''. EX begins by using a greedy algorithm to find a solution that meets the desired threshold. If this greedy choice fails, EX begins a search of all feasible sets where each element considered is in order of decreasing marginal gains. Lazy updates are used for computing marginal gains. ``F-EX'' uses a similar approach as EX, but as described in the main paper limits the search to a smaller portion of the elements. In particular, as described in the main paper, the F-EX algorithm first checks if the optimization problem in line \ref{alg:exact_step} of \filt is an unconstrained \smp. If so, the algorithm adds all monotone elements from the the union of $\{S_1,...S_{2/\epsilon}\}$ to the solution set and then explores the non-monotone elements in the rest of $\cup_{i=1}^{2/\epsilon}S_i$. 

\subsection{Additional experimental results}
\label{appendix:results}
First of all, we present some additional monotone experimental results on the synthetic data and the corel dataset. The corel dataset is the Corel5k set of
images in \cite{duygulu2002object} ($n=4500$). The synthetic data we used here is a data summarization instance with a total of $m=4000$ elements and 2000 ($n=2000$) subset of elements. Each subset is randomly generated in the following way: let us denote the set of elements as $[m]=\{1,2,3...,m\}$, each elements $i\leq250$ is added to the subset with probability $0.4$, and each element $250<i\leq m$ is added with probability $0.002$. The four algorithms examined here can be found in Section \ref{section:exp}. We compare the four algorithms for different value of $\tau$ and $\epsilon$. When $\epsilon$ is varied, $\tau$ is fixed at $0.9f(U)$ for corel dataset and synthetic dataset with $U$ being the universe of the two instances respectively. When $\tau$ is varied, $\epsilon$ is fixed at  $0.05$. The results on the synthetic data are presented in Figure \ref{fig:synthetic_tau_f}, \ref{fig:synthetic_tau_c}, \ref{fig:synthetic_tau_q}, and \ref{fig:synthetic_eps_q}. The results on the corel dataset are plotted in Figure \ref{fig:corel_tau_f}, \ref{fig:corel_tau_c}, \ref{fig:corel_tau_q} and \ref{fig:corel_eps_q}.From the results, we can see that the results on the corel dataset is in line with the results in the main paper. However, on the synthetic dataset, it is worth noting that the "SG" (\sgc) algorithm requires less number of queries compared with other three algorithms even when $\epsilon$ is large, which further demonstrates the advantages of our algorithms.

\begin{figure*}[t!]
    \centering
\hspace{-0.5em}
\subfigure[synthetic $f$]
{\label{fig:synthetic_tau_f}\includegraphics[width=0.24\textwidth]{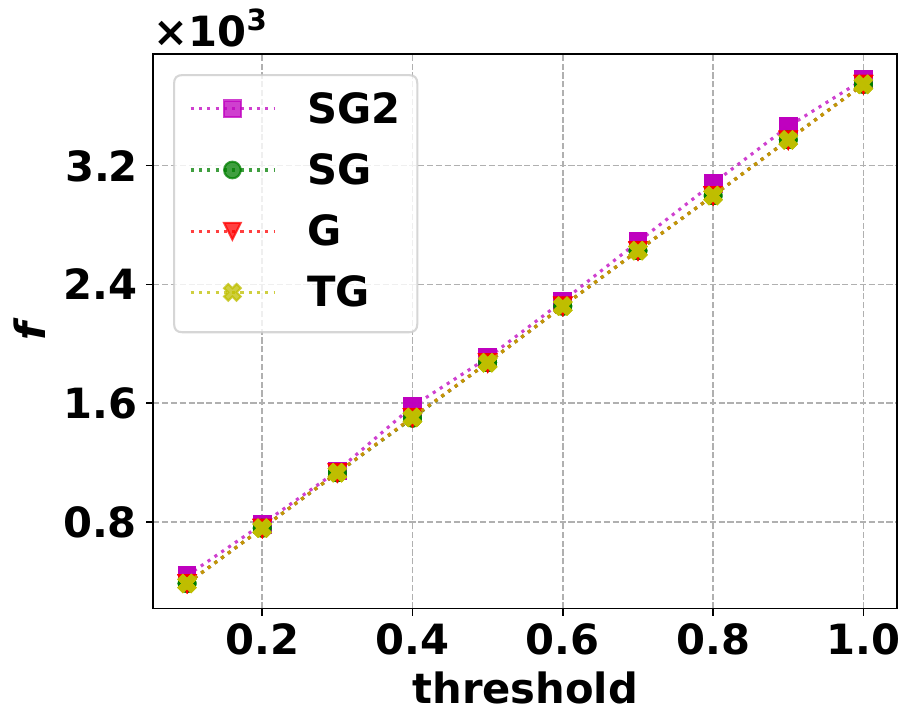}}
\hspace{-0.5em}
\subfigure[synthetic $c$]
{\label{fig:synthetic_tau_c}\includegraphics[width=0.24\textwidth]{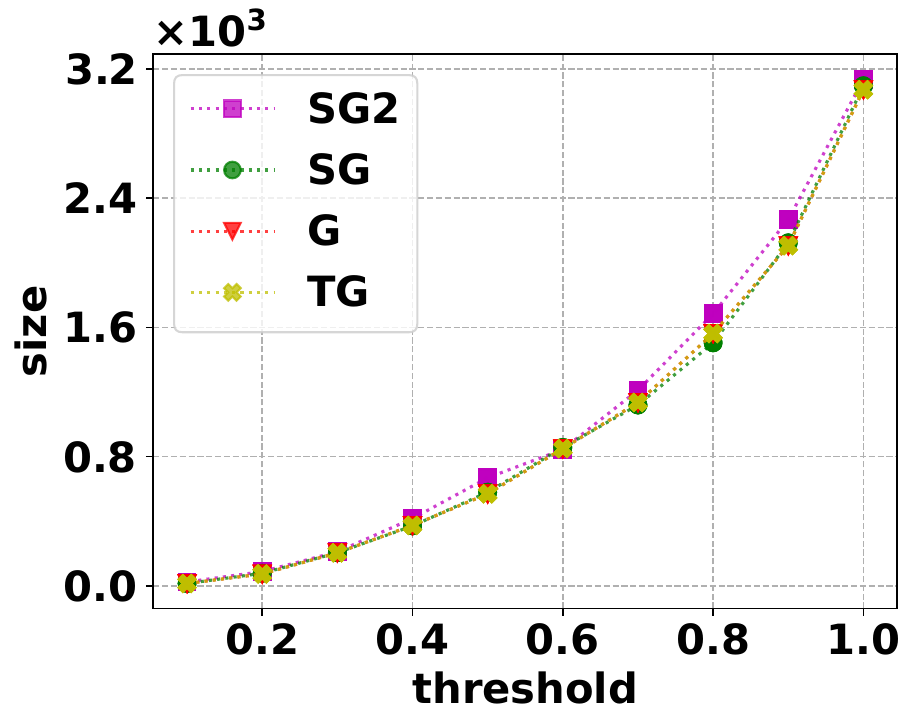}}
\hspace{-0.5em}
\subfigure[synthetic queries]
{\label{fig:synthetic_tau_q}\includegraphics[width=0.24\textwidth]{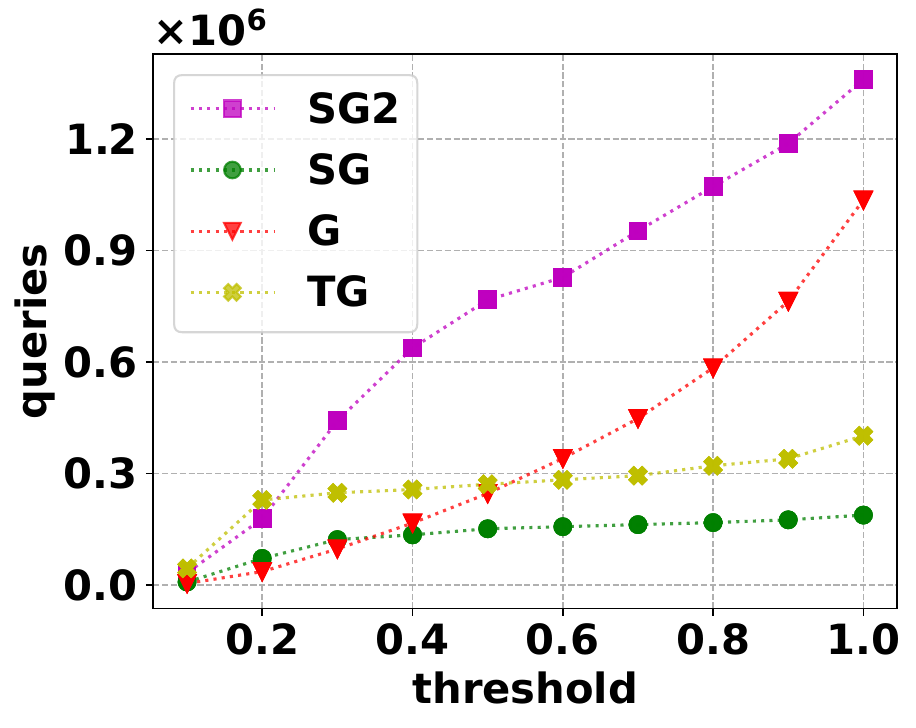}}
\hspace{-0.5em}
\subfigure[synthetic queries]
{\label{fig:synthetic_eps_q}\includegraphics[width=0.24\textwidth]{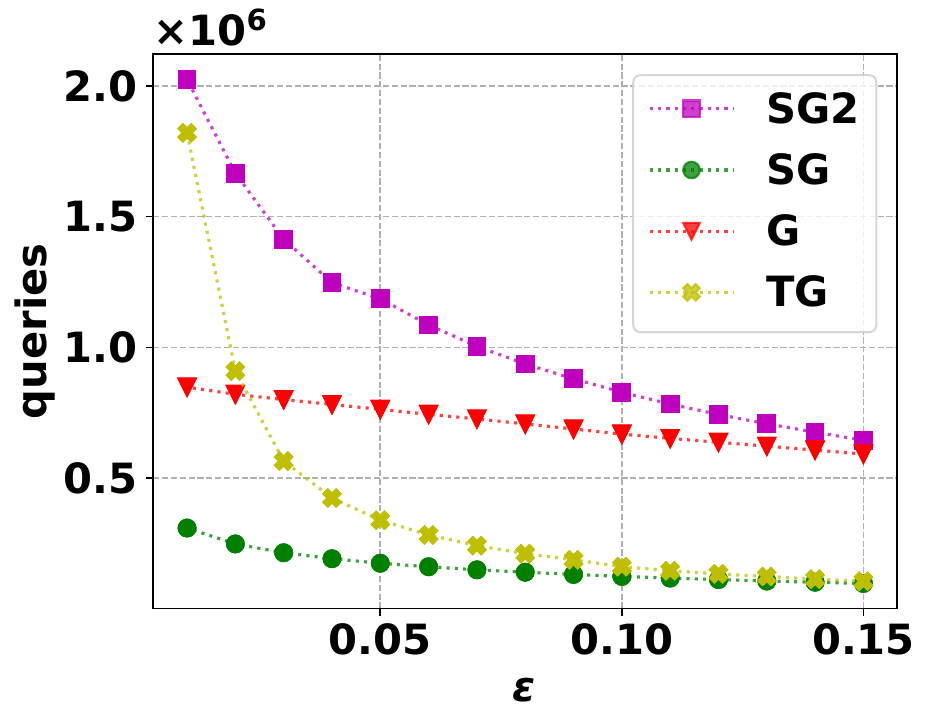}}
\hspace{-0.5em}
\subfigure[corel $f$]
{\label{fig:corel_tau_f}\includegraphics[width=0.24\textwidth]{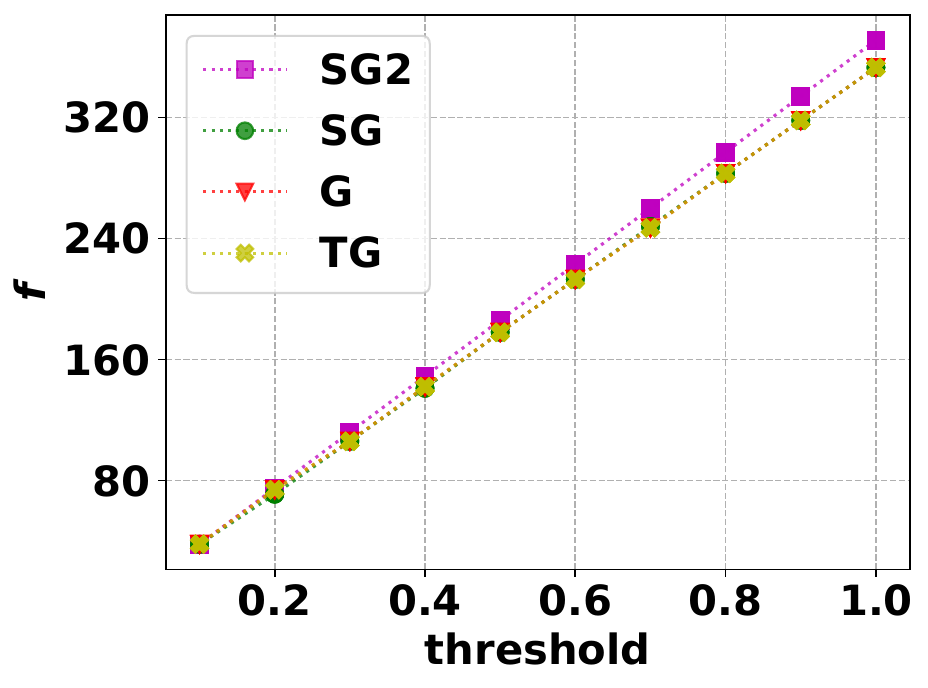}}
\hspace{-0.5em}
\subfigure[corel $c$]
{\label{fig:corel_tau_c}\includegraphics[width=0.24\textwidth]{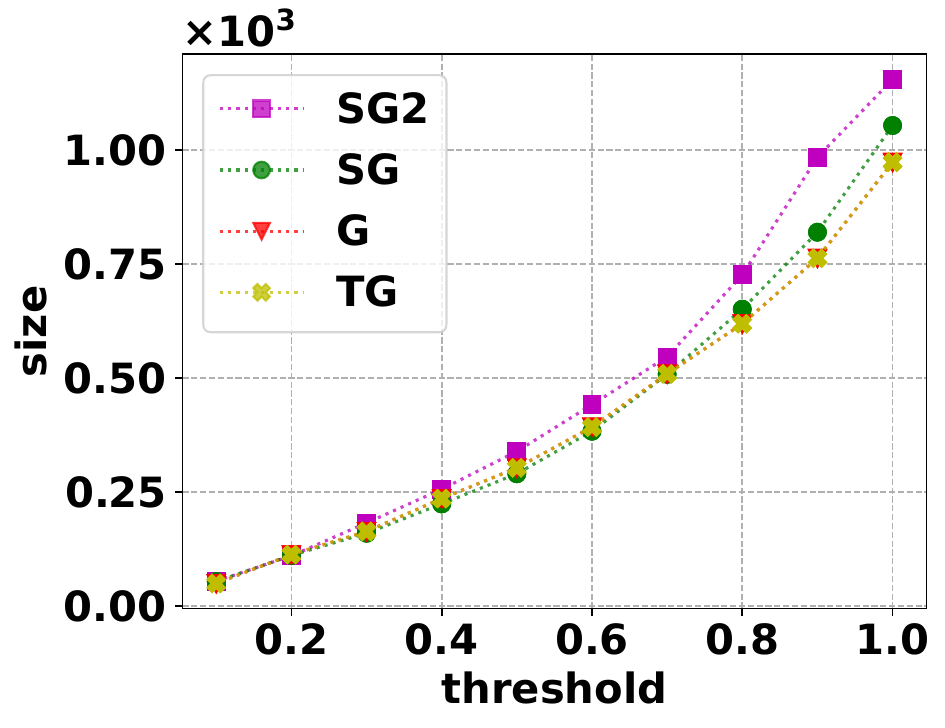}}
\hspace{-0.5em}
\subfigure[corel queries]
{\label{fig:corel_tau_q}\includegraphics[width=0.24\textwidth]{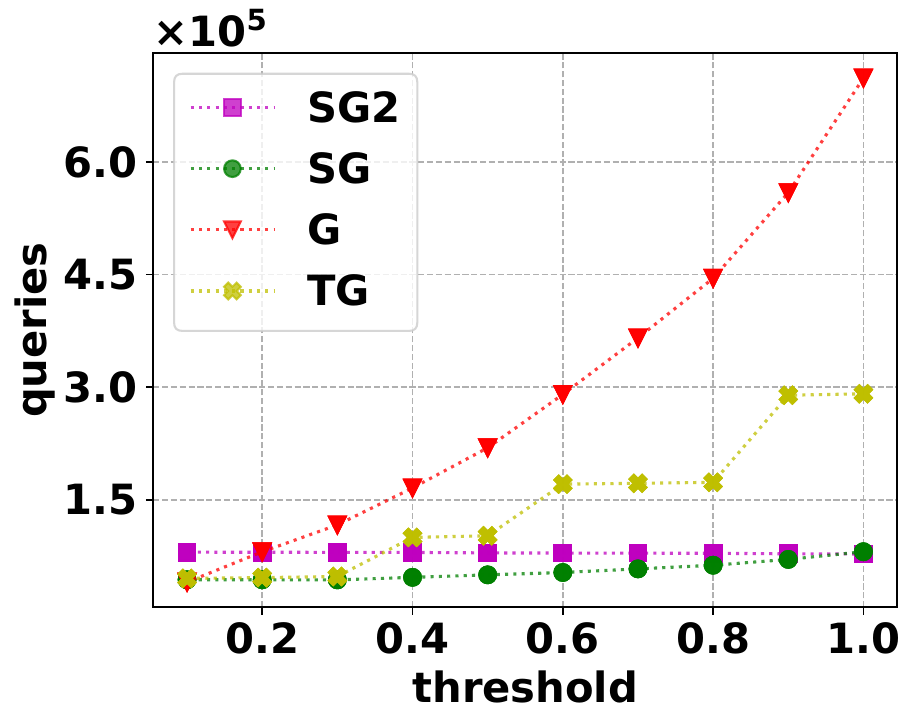}}
\hspace{-0.5em}
\subfigure[corel queries]
{\label{fig:corel_eps_q}\includegraphics[width=0.24\textwidth]{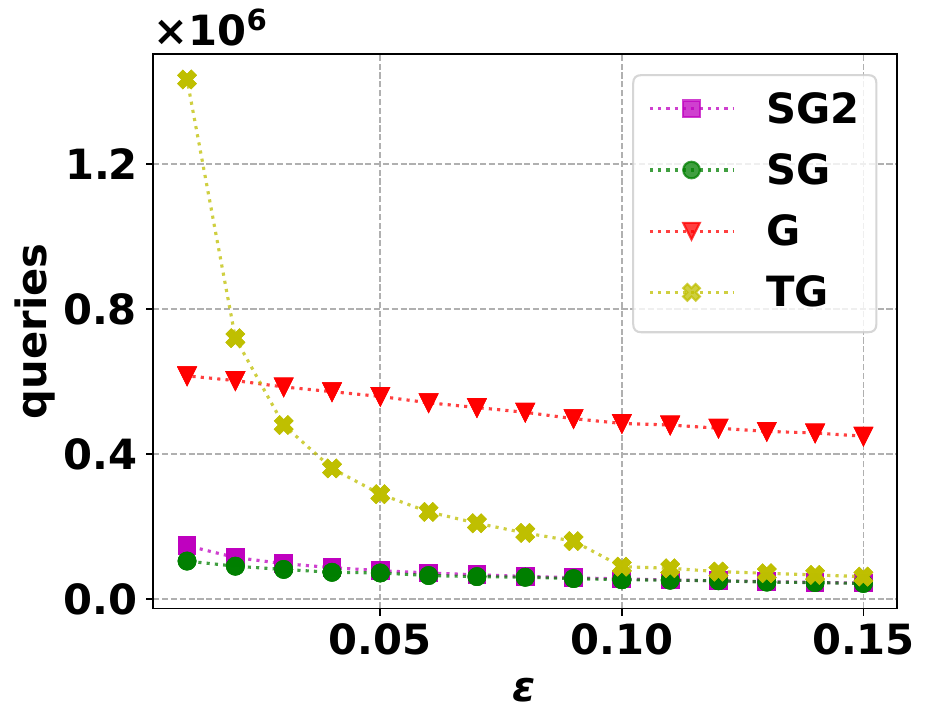}}
\caption{The experimental results of running different greedy algorithms on the instances of data summarization on the synthetic dataset, and corel dataset.}
        \label{fig:mono_synthetic}
\end{figure*}

The additional experiments we present are further exploration of our algorithm \filt.
In Figure \ref{fig:nonmono_filter}, we present additional experimental results analogous to those presented in the main paper for graph cut, but on the additional datasets described above. In summary, we see many of the same patterns exhibited on these instances as discussed in the main paper. In addition, we include an additional set of experiments analyzing how many non-monotone elements there are in the union of $\{S_1,...S_{2/\epsilon}\}$ in \filt in order to analyze how effective we would expect F-EX to be. The number (``num'') and portion (``pt'') of the nonmonotone elements the union of $\{S_1,...S_{2/\epsilon}\}$ on the instances with the graph cut objective are plotted in Figures \ref{fig:nonmono_eps} and \ref{fig:nonmono_tau}. If that instance did not require an exact search (meaning that the initial greedy heuristic found a solution), then -1 is plotted. From the figures, we can see that in most cases, the number of nonmonotone elements is either $0$ or very small, which explains why the fast exact algorithm requires less queries than the exact algorithm in these cases. This implies that in many instances of \scp, \filt is able to cut down the original instance to one that is nearly monotone and much easier to solve.

\begin{figure*}[t!]
    \centering
\hspace{-0.5em}
\subfigure[amazon $f$]
{\label{fig:amazon_tau_f}\includegraphics[width=0.24\textwidth]{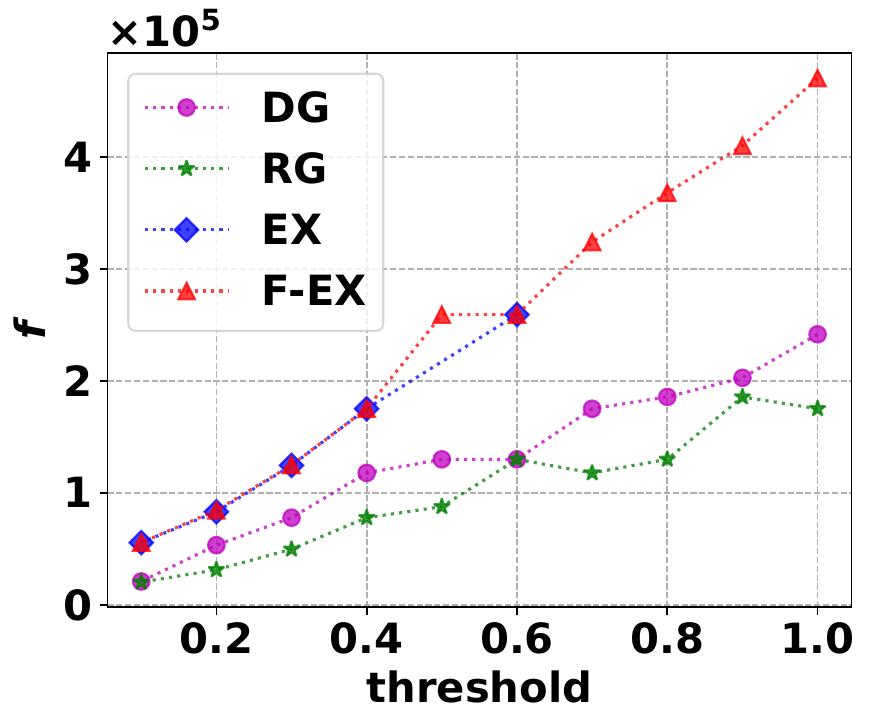}}
\hspace{-0.5em}
\subfigure[amazon $c$]
{\label{fig:amazon_tau_c}\includegraphics[width=0.24\textwidth]{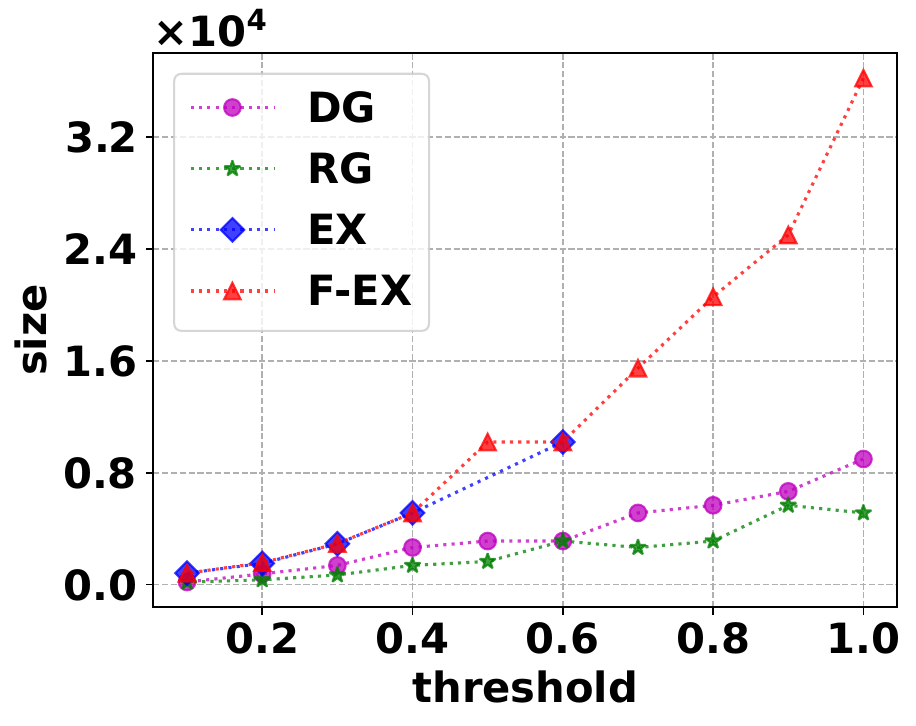}}
\hspace{-0.5em}
\subfigure[amazon queries]
{\label{fig:amazon_tau_q}\includegraphics[width=0.24\textwidth]{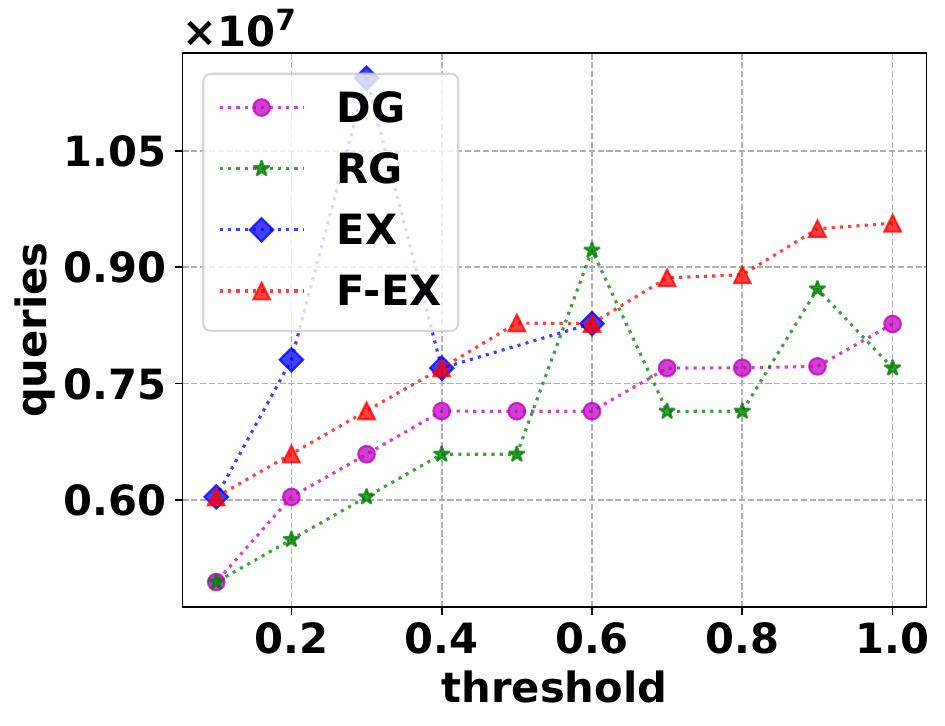}}
\hspace{-0.5em}
\subfigure[amazon queries]
{\label{fig:amazon_eps_q}\includegraphics[width=0.24\textwidth]{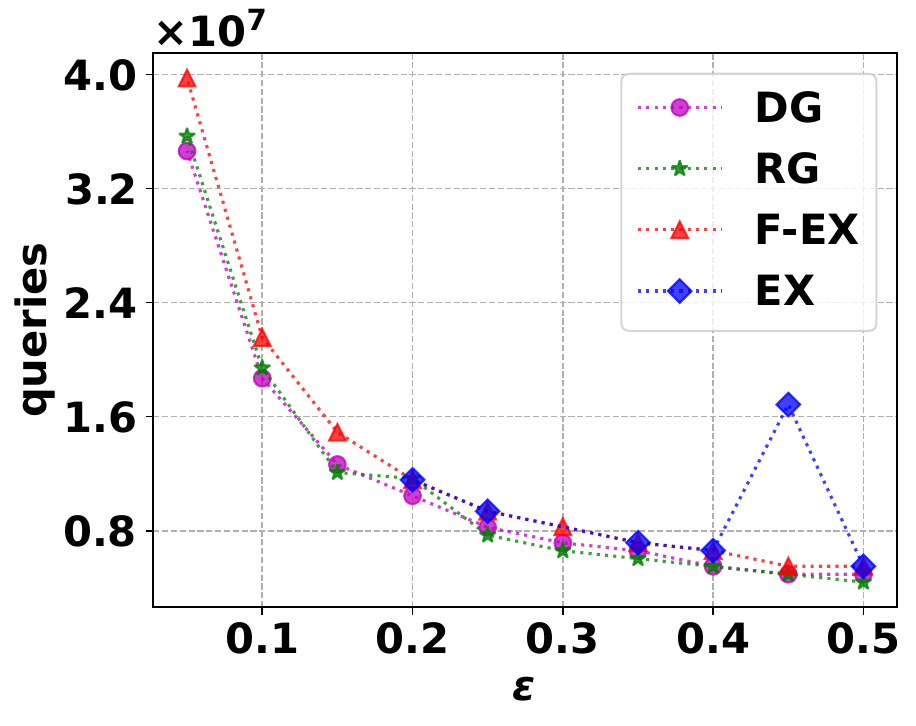}}
\hspace{-0.5em}
\subfigure[enron $f$]
{\label{fig:enron_tau_f}\includegraphics[width=0.24\textwidth]{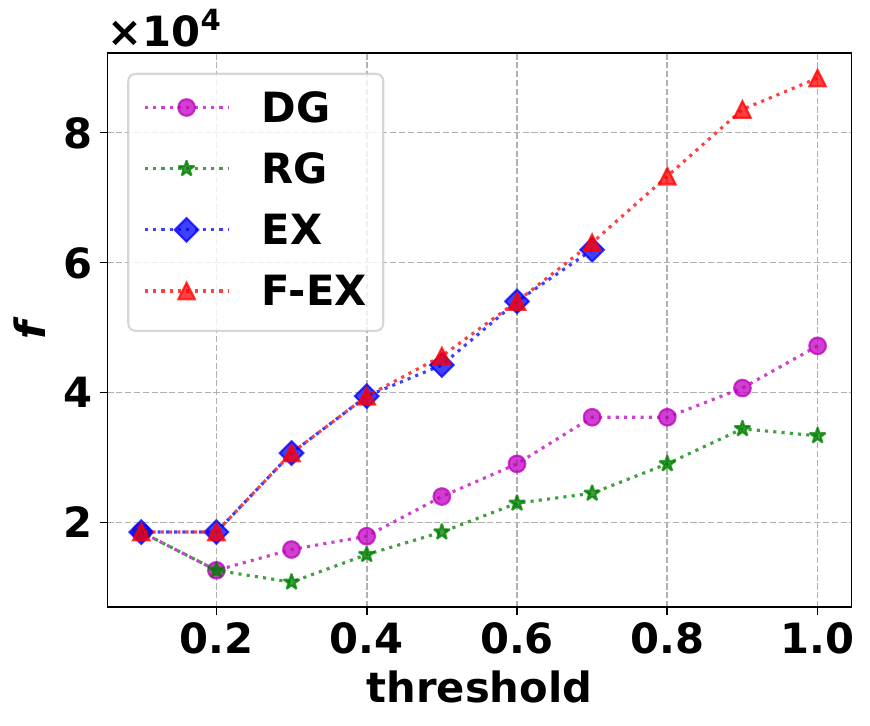}}
\hspace{-0.5em}
\subfigure[enron $c$]
{\label{fig:enron_tau_c}\includegraphics[width=0.24\textwidth]{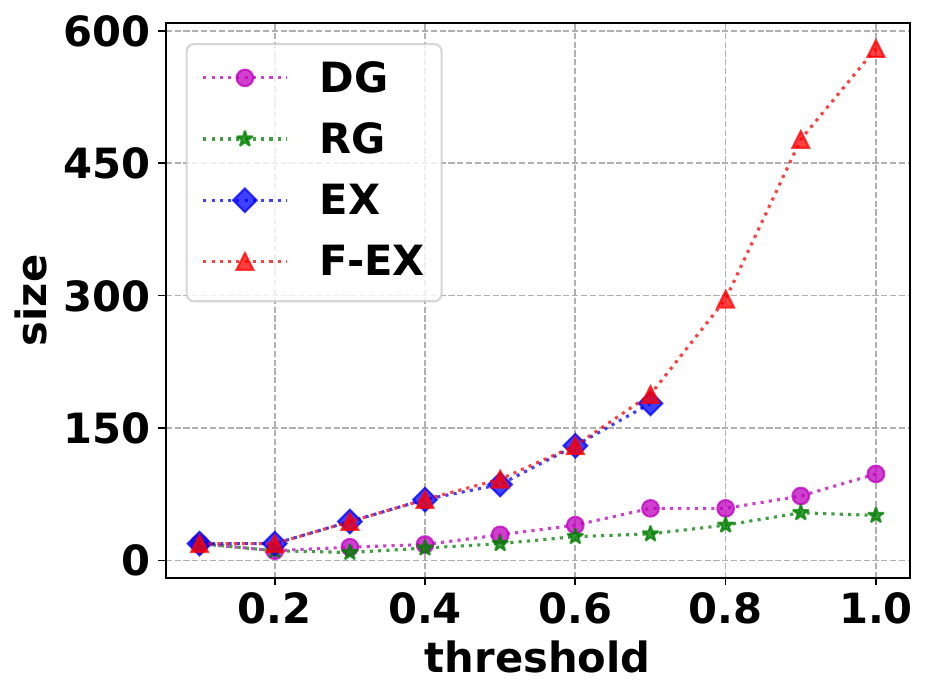}}
\hspace{-0.5em}
\subfigure[enron queries]
{\label{fig:enron_tau_q}\includegraphics[width=0.24\textwidth]{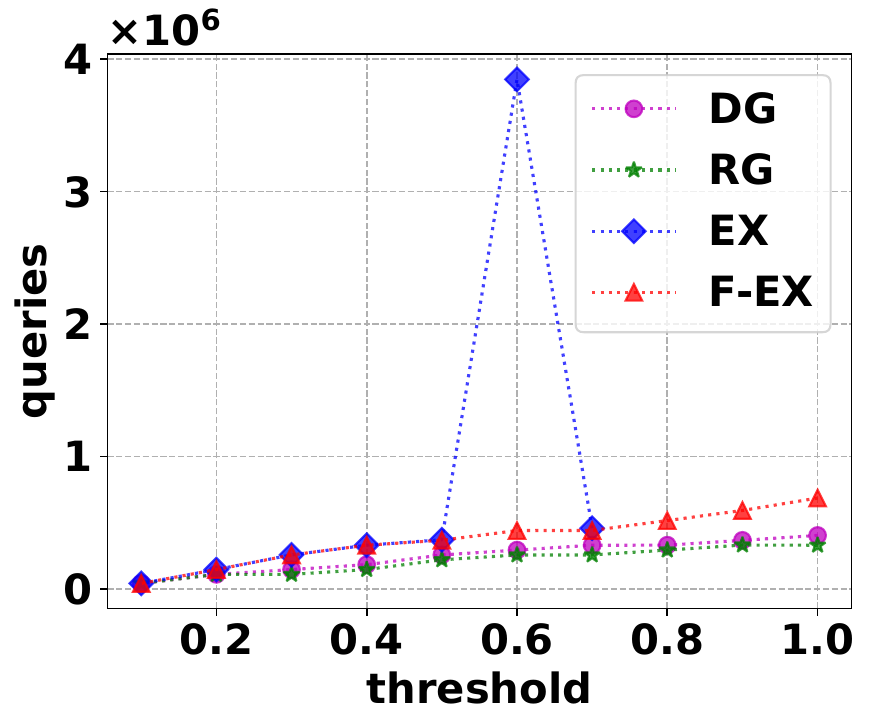}}
\hspace{-0.5em}
\subfigure[enron queries]
{\label{fig:enron_eps_q}\includegraphics[width=0.24\textwidth]{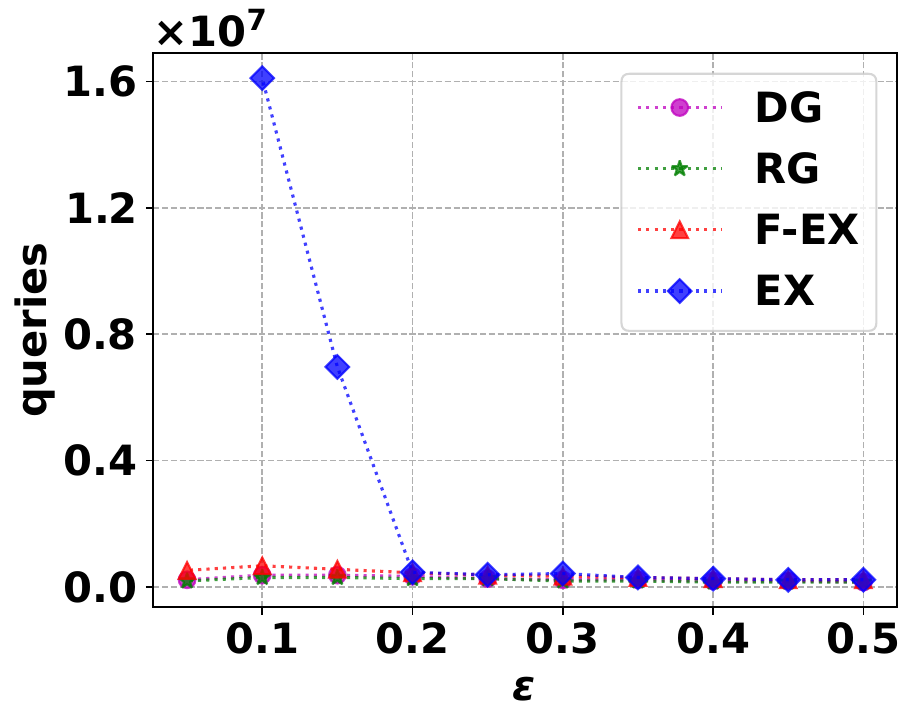}}
\hspace{-0.5em}
\subfigure[facebook $f$]
{\label{fig:facebook_tau_f}\includegraphics[width=0.24\textwidth]{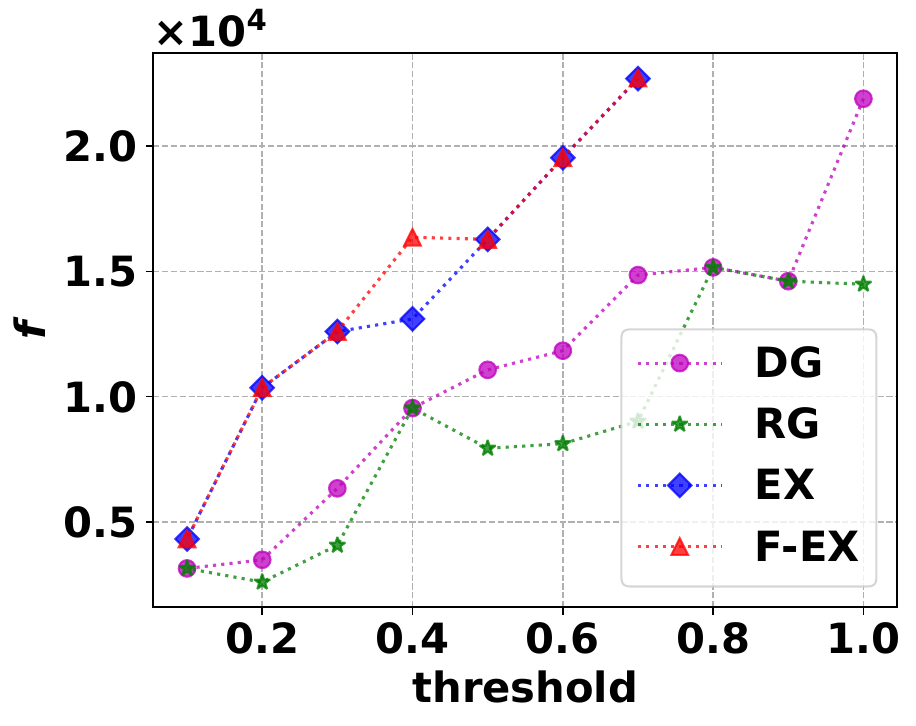}}
\hspace{-0.5em}
\subfigure[facebook $c$]
{\label{fig:facebook_tau_c}\includegraphics[width=0.24\textwidth]{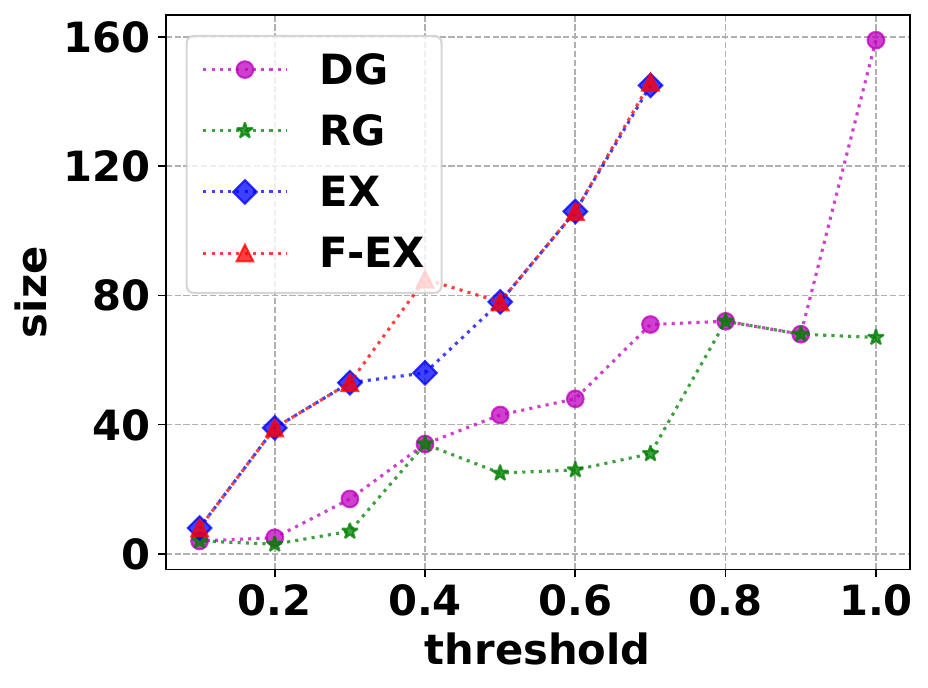}}
\hspace{-0.5em}
\subfigure[facebook queries]
{\label{fig:facebook_tau_q}\includegraphics[width=0.24\textwidth]{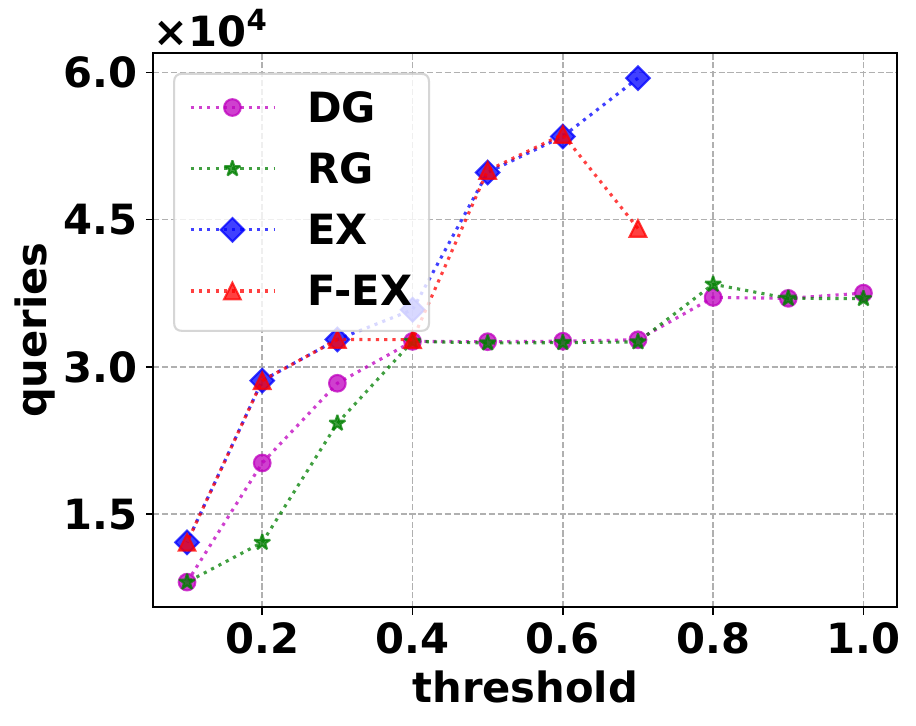}}
\hspace{-0.5em}
\subfigure[facebook queries]
{\label{fig:facebook_eps_q}\includegraphics[width=0.24\textwidth]{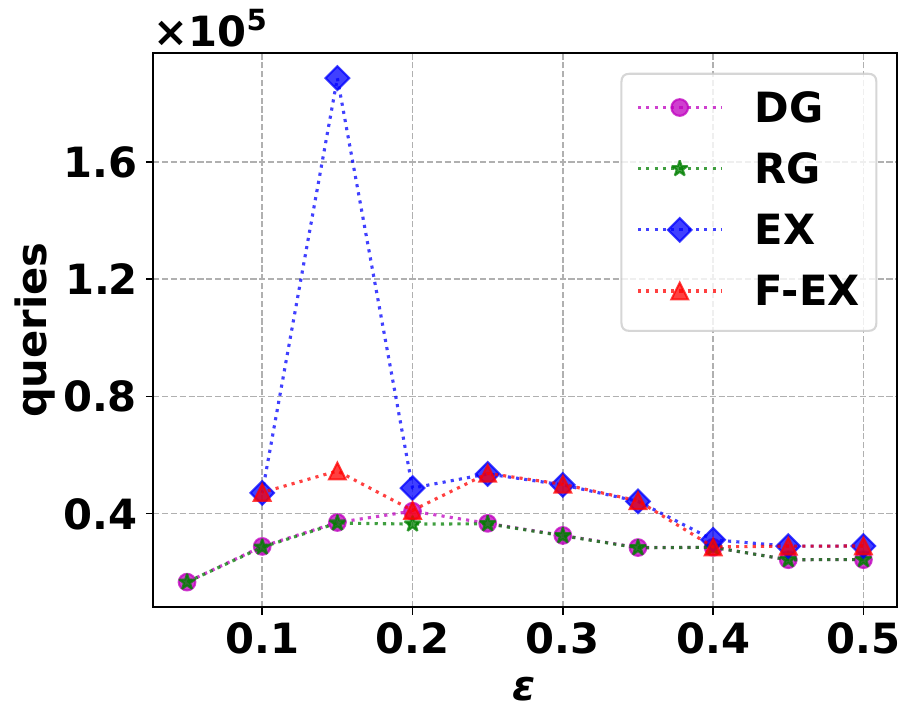}}
\caption{The experimental results of running \filt on the instances of graph cut on the com-Amazon graph ("amazon"), email-Enron ("enron") and ego-Facebook ("facebook") dataset.}
        \label{fig:nonmono_filter}
\end{figure*}

\begin{figure}[t]
    \centering
    \subfigure[amazon, $\tau=0.5\cdot f_{\text{ave-max}}$]{\label{fig:amazon_epsnum}\includegraphics[width=0.23\textwidth]{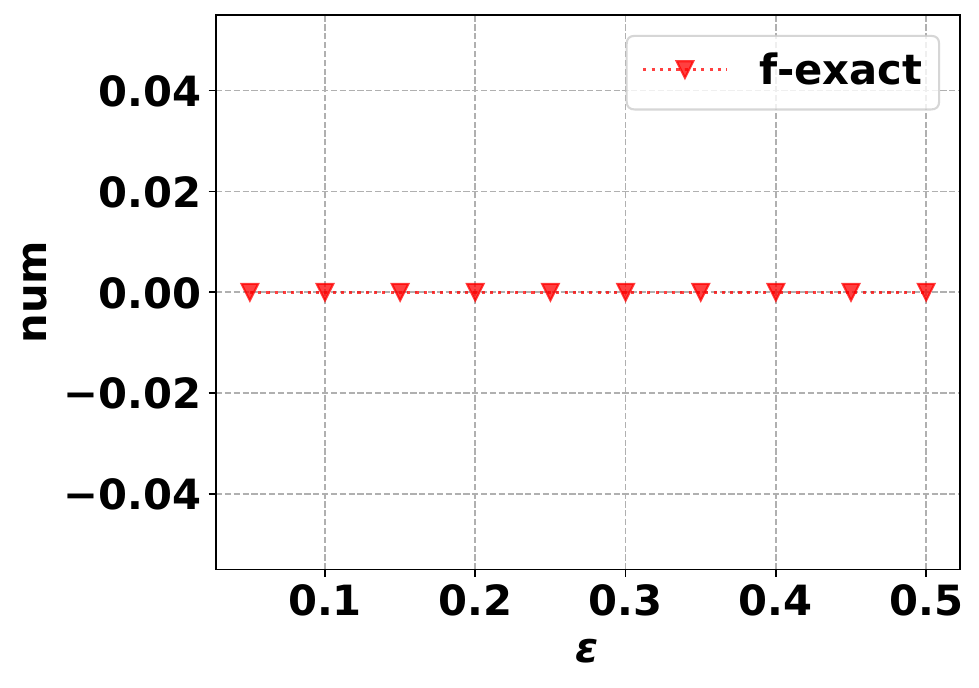}}
    \subfigure[facebook, $\tau=0.5\cdot f_{\text{ave-max}}$]{\label{fig:facebook_epsnum}\includegraphics[width=0.23\textwidth]{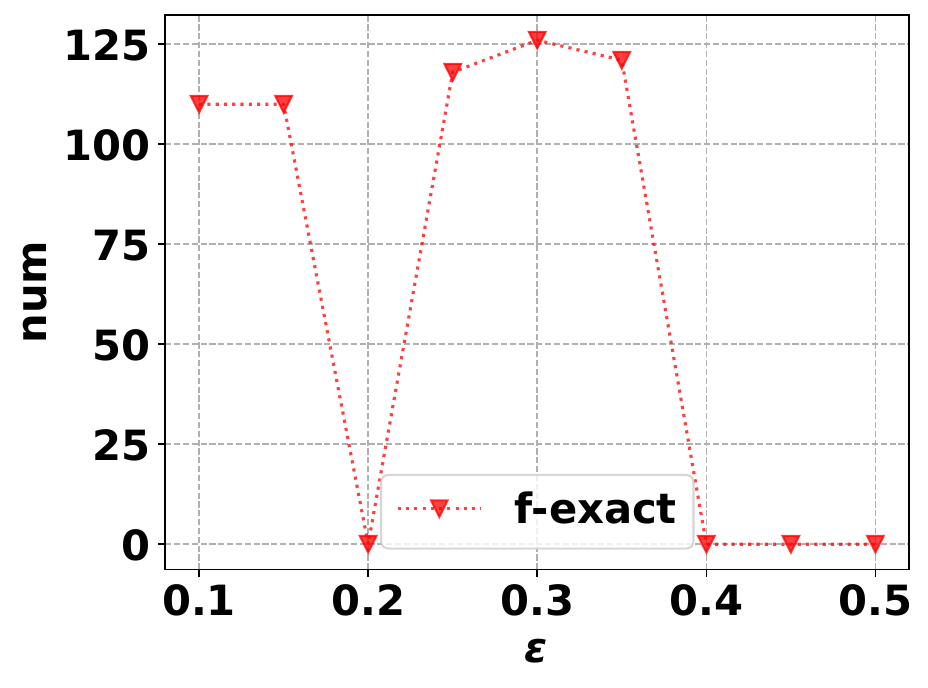}}
    \subfigure[enron, $\tau=0.7\cdot f_{\text{ave-max}}$]{\label{fig:enron_epsnum}\includegraphics[width=0.23\textwidth]{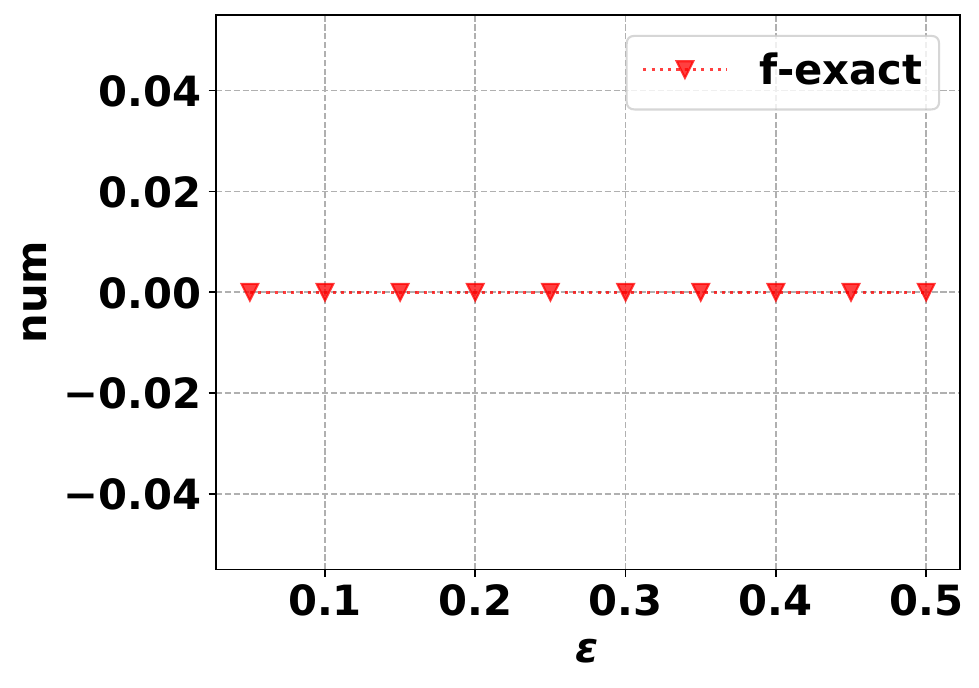}}
    \subfigure[euall, $\tau=0.9\cdot f_{\text{ave-max}}$]{\label{fig:euall_epsnum}\includegraphics[width=0.23\textwidth]{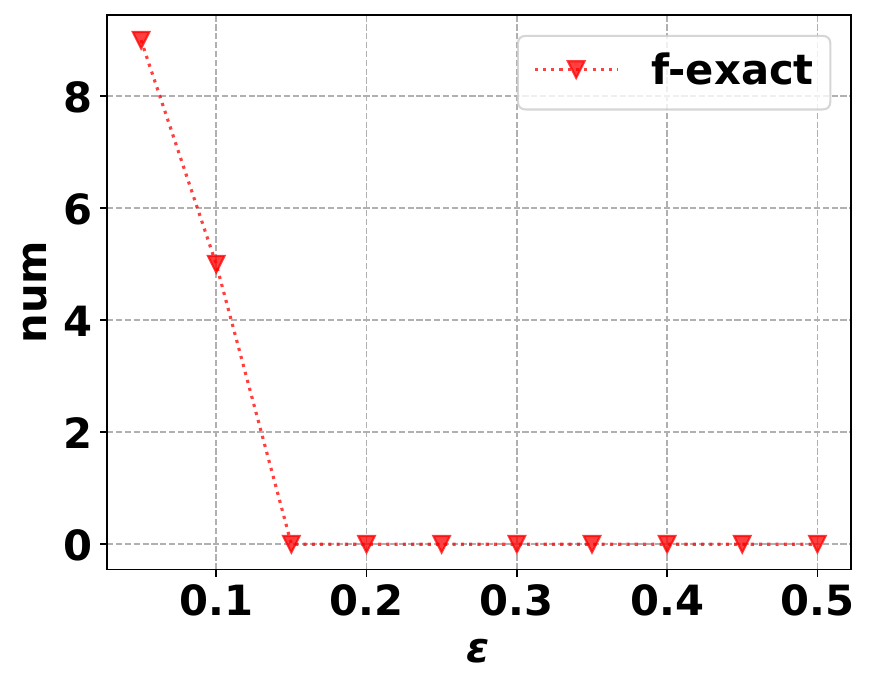}}
    \vspace{\baselineskip} 

    \subfigure[amazon, $\tau=0.5\cdot f_{\text{ave-max}}$]{\label{fig:amazon_epspt}\includegraphics[width=0.23\textwidth]{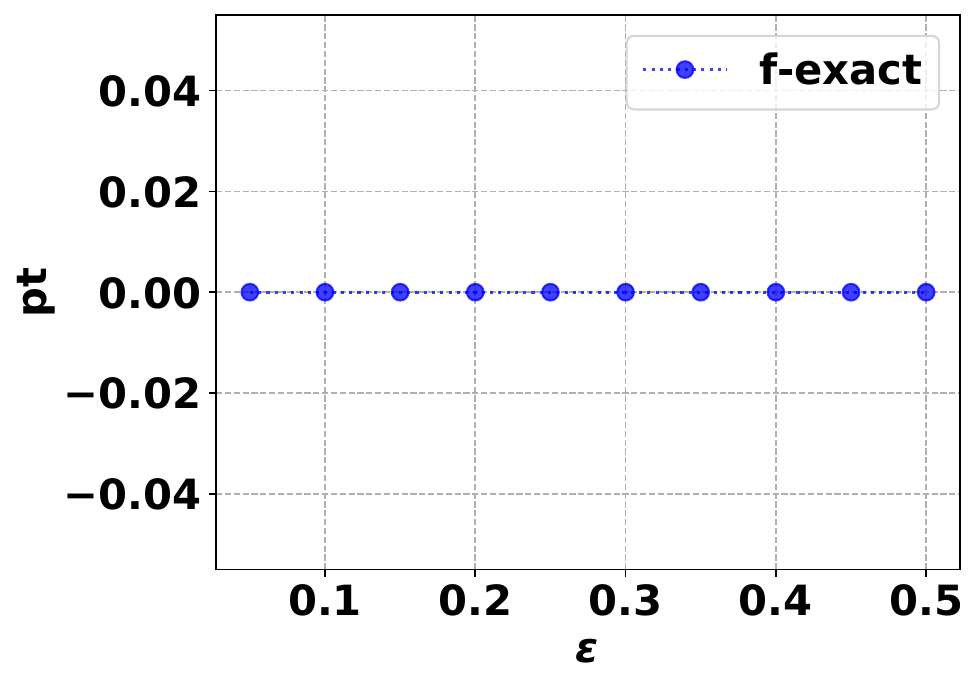}}
    \subfigure[facebook, $\tau=0.5\cdot f_{\text{ave-max}}$]{\label{fig:facebook_epspt}\includegraphics[width=0.23\textwidth]{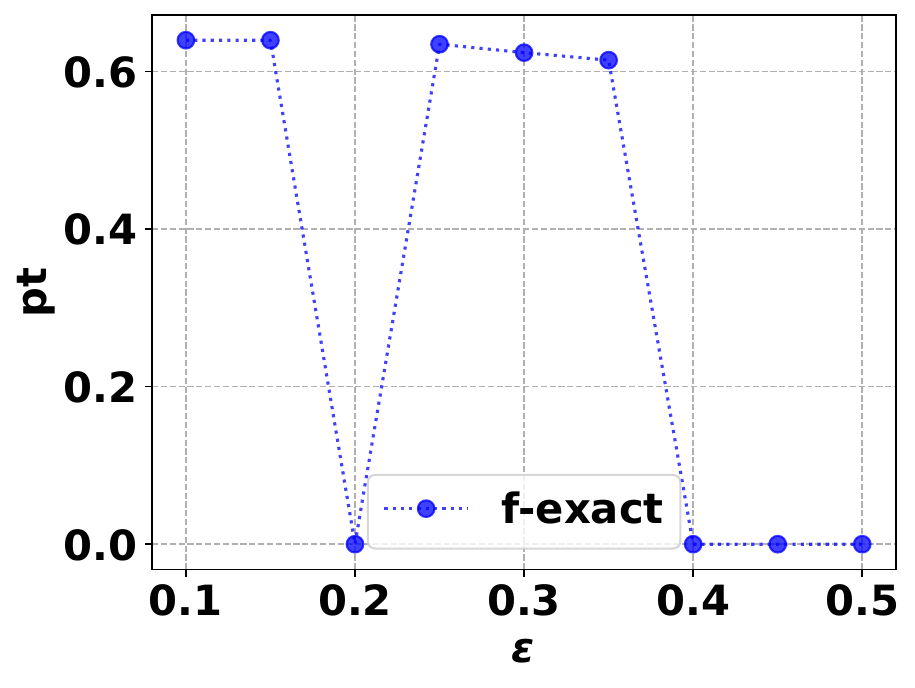}}
    \subfigure[enron, $\tau=0.7\cdot f_{\text{ave-max}}$]{\label{fig:enron_epspt}\includegraphics[width=0.23\textwidth]{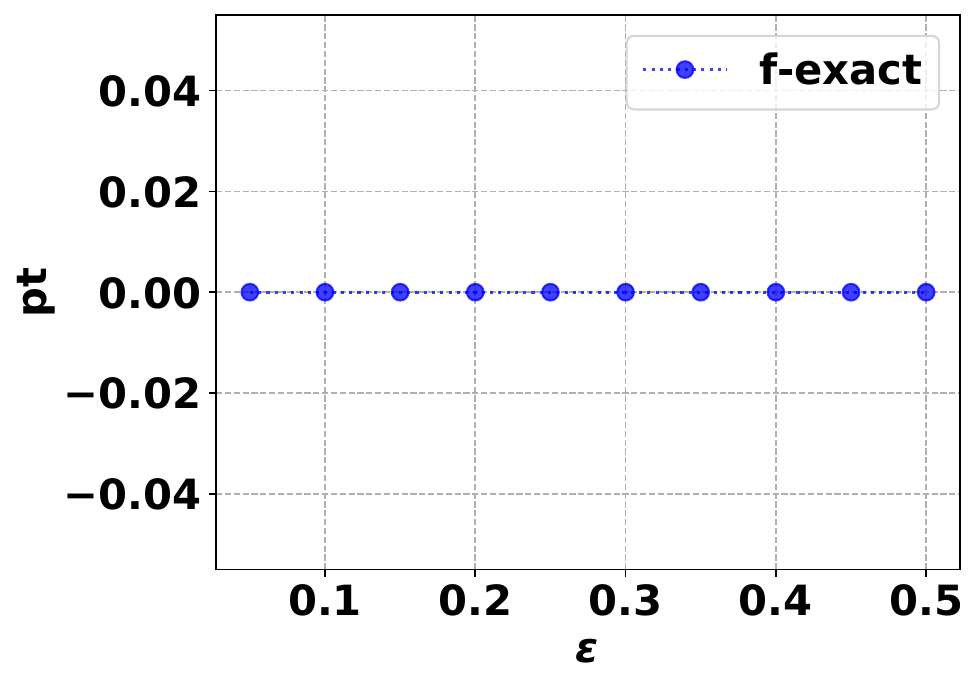}}
    \subfigure[euall, $\tau=0.9\cdot f_{\text{ave-max}}$]{\label{fig:euall_epspt}\includegraphics[width=0.23\textwidth]{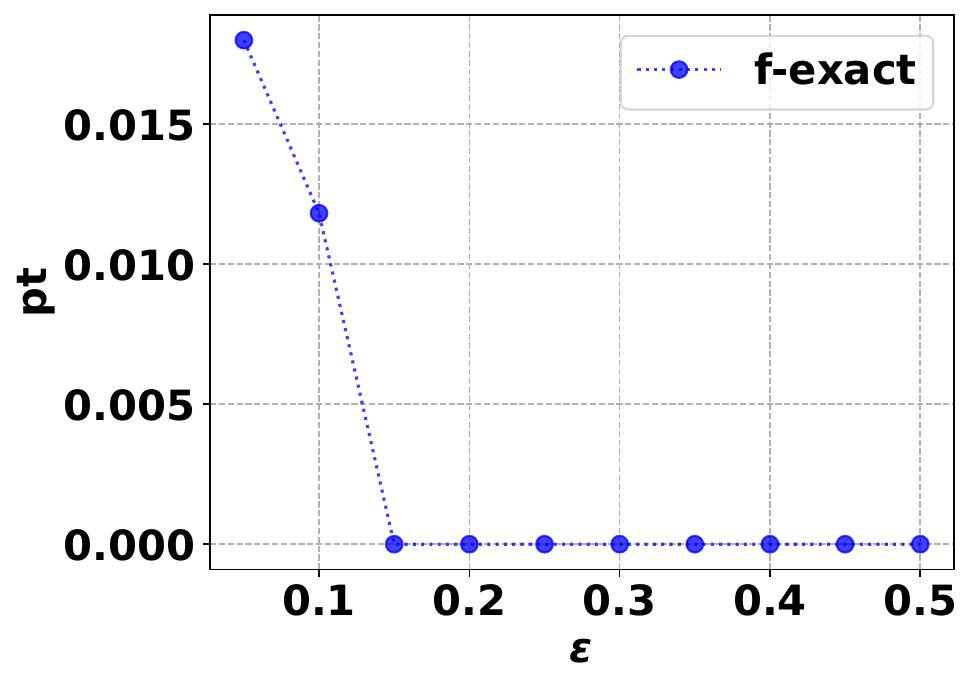}}
    
    \caption{The second experiments. We plot the number and portion of the non-monotone elements. Here $x$-axis refers to value of $\tau$. 
    }
    \label{fig:nonmono_eps}
\end{figure}

\begin{figure}[t]
    \centering
    \subfigure[amazon, $\epsilon=0.3$]{\label{fig:amazon_taunum}\includegraphics[width=0.23\textwidth]{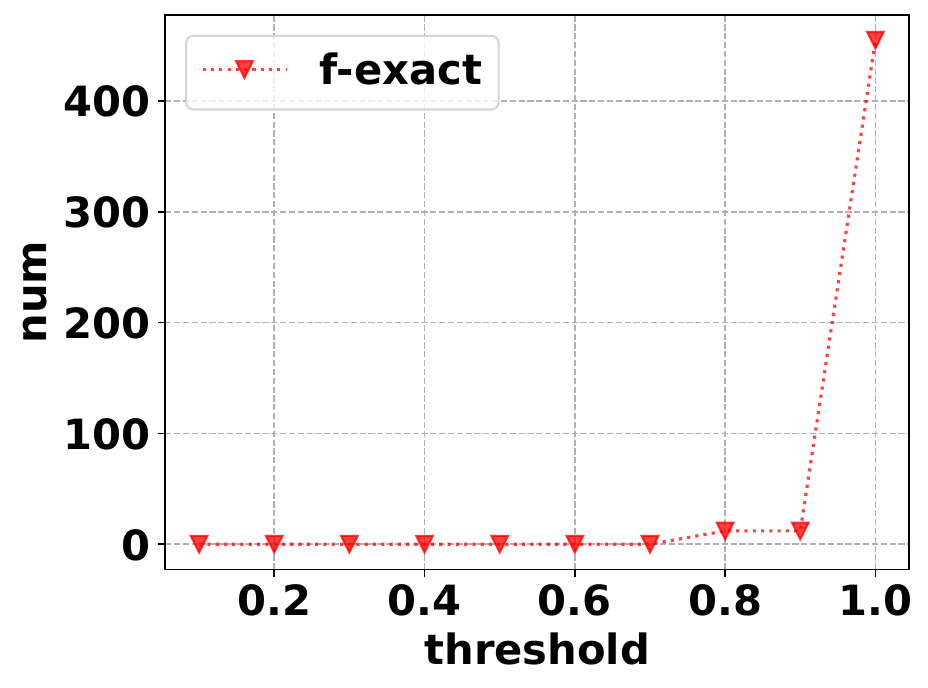}}
    \subfigure[facebook, $\epsilon=0.3$]{\label{fig:facebook_taunum}\includegraphics[width=0.23\textwidth]{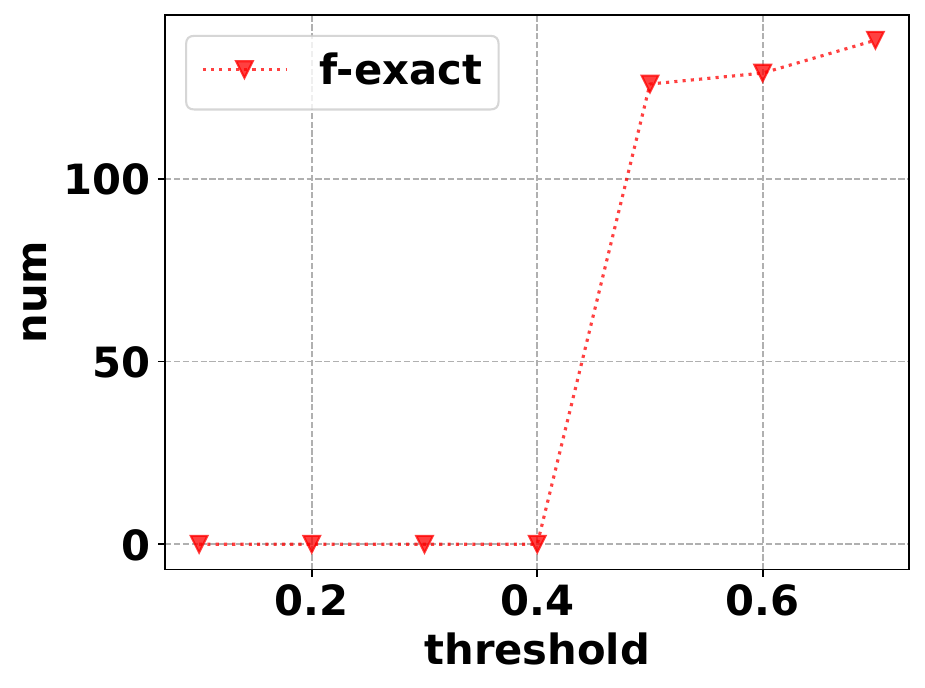}}
    \subfigure[enron, $\epsilon=0.2$]{\label{fig:enron_taunum}\includegraphics[width=0.23\textwidth]{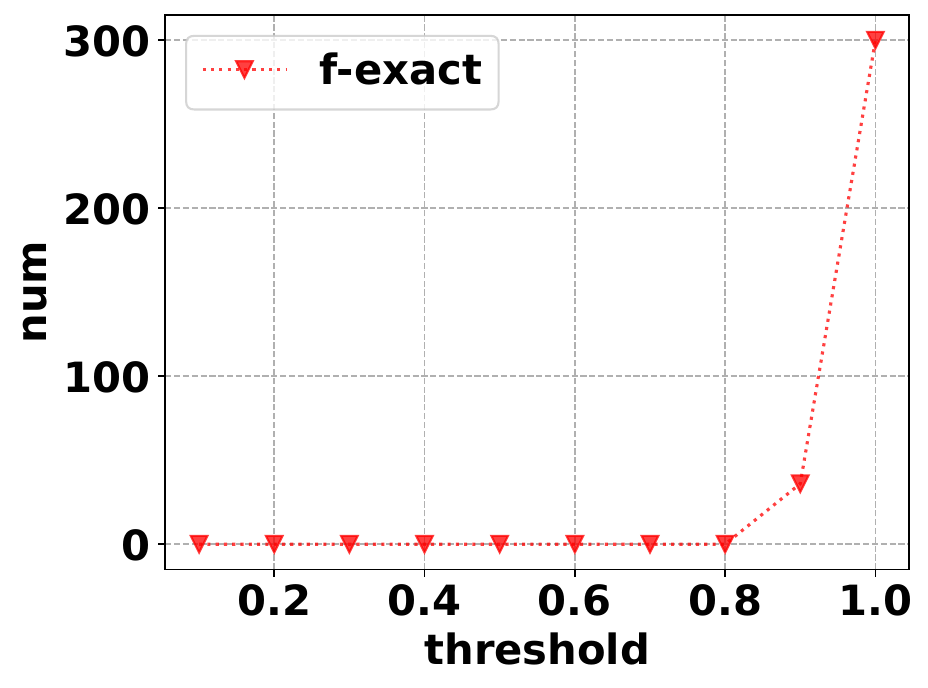}}
    \subfigure[euall, $\epsilon=0.15$]{\label{fig:euall_taunum}\includegraphics[width=0.23\textwidth]{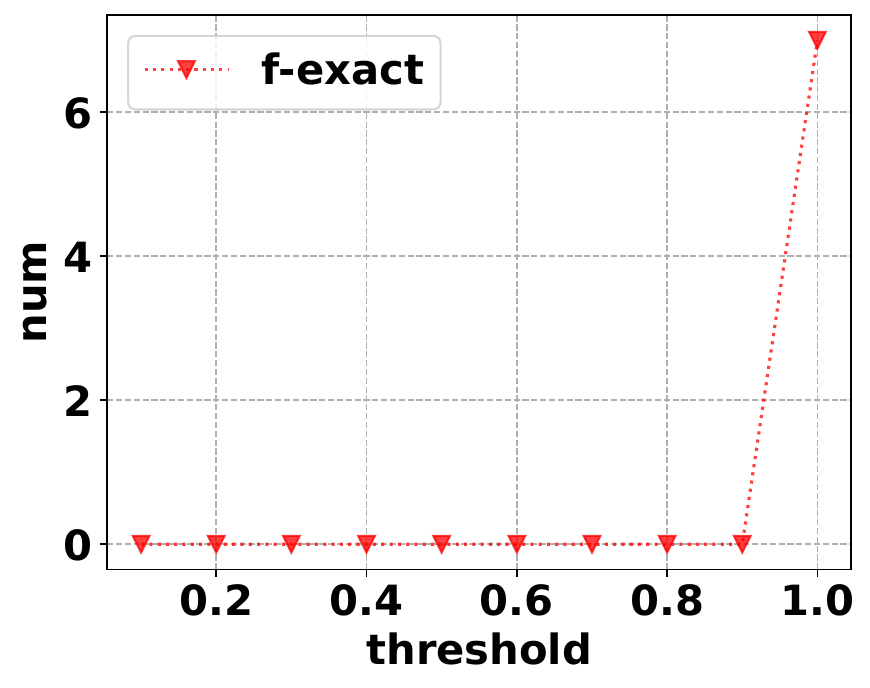}}
    \vspace{\baselineskip} 

    \subfigure[amazon, $\epsilon=0.3$]{\label{fig:amazon_taupt}\includegraphics[width=0.23\textwidth]{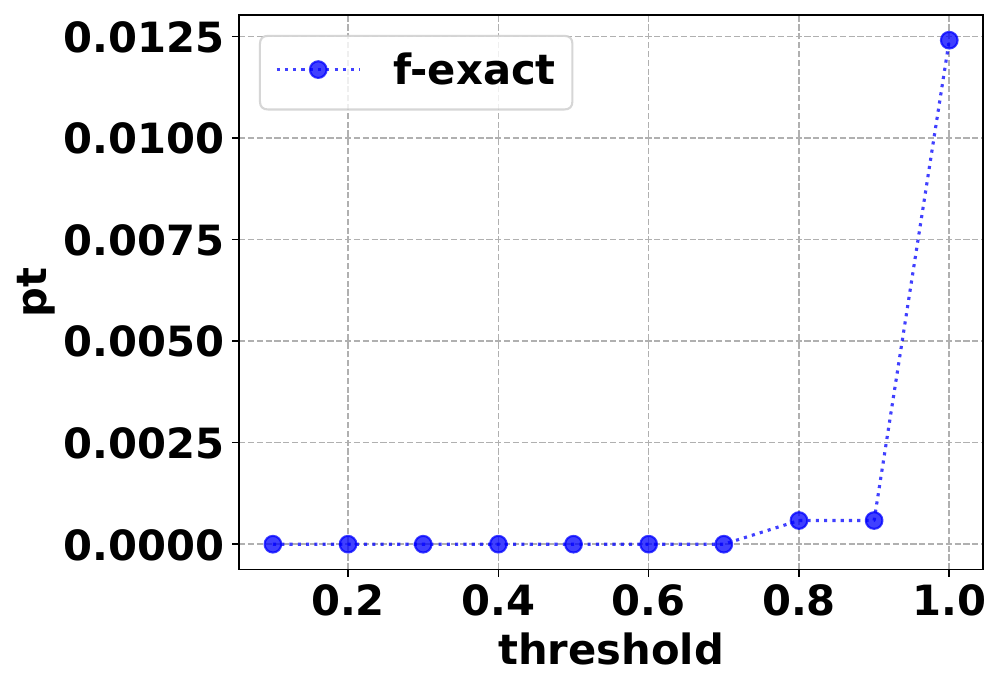}}
     \subfigure[facebook, $\epsilon=0.3$]{\label{fig:facebook_taupt}\includegraphics[width=0.23\textwidth]{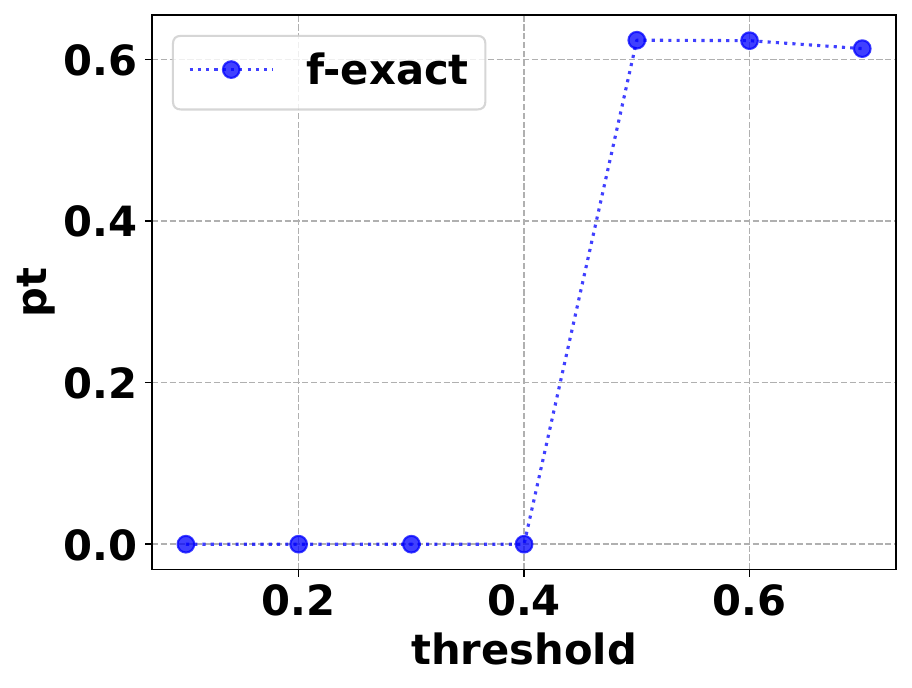}}
   \subfigure[enron, $\epsilon=0.2$]{\label{fig:enron_taupt}\includegraphics[width=0.23\textwidth]{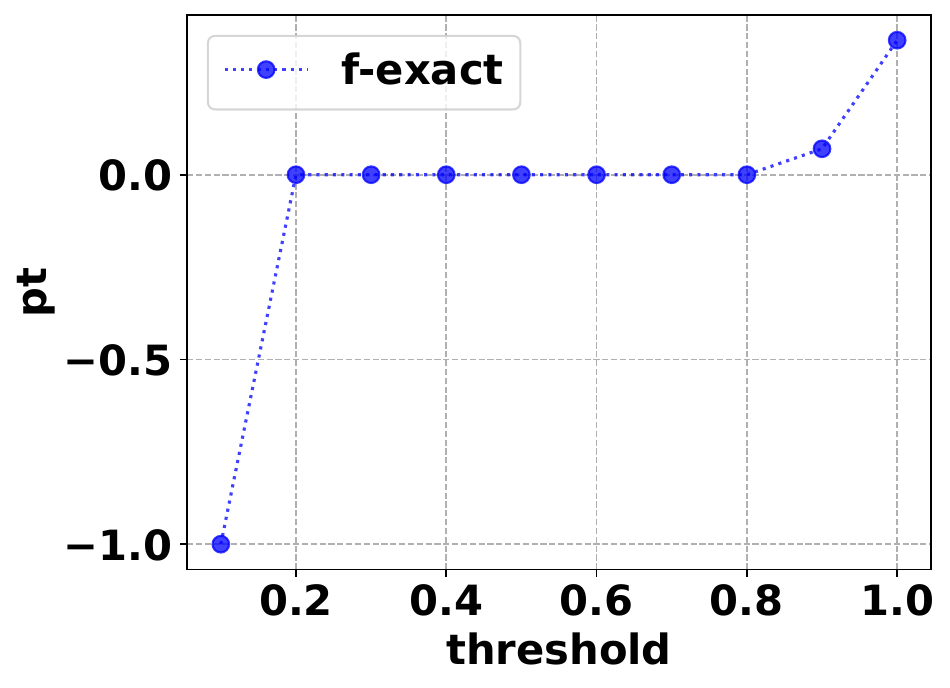}}
    \subfigure[euall, $\epsilon=0.15$]{\label{fig:euall_taupt}\includegraphics[width=0.23\textwidth]{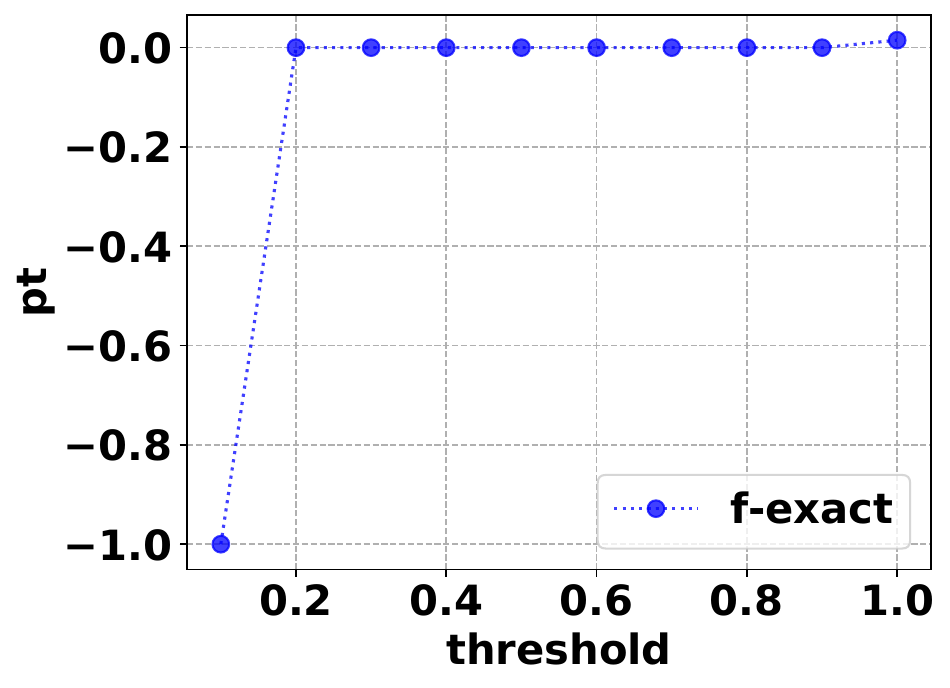}}
    
    \caption{The second experiments. We plot the number and portion of the non-monotone elements.
Here $x$-axis refers to value of $\epsilon$.     }
    \label{fig:nonmono_tau}
\end{figure}

    

\end{document}